\documentclass[a4paper,UKenglish,cleveref, autoref, thm-restate]{lipics-v2021}

\hideLIPIcs 

\pdfoutput=1 

\title{A practical algorithm for $2$-admissibility} 

\author{Christine Awofeso}{Birkbeck, University of London, UK}%
{cawofe01@student.bbk.ac.uk}
{}
{}

\author{Patrick Greaves}{Birkbeck, University of London, UK}%
{p.greaves@bbk.ac.uk}
{}%
{}

\author{Oded Lachish}{Birkbeck, University of London, UK}%
{o.lachish@bbk.ac.uk}%
{https://orcid.org/0000-0001-5406-8121}
{}

\author{Felix Reidl}{Birkbeck, University of London, UK}%
{f.reidl@bbk.ac.uk}%
{https://orcid.org/0000-0002-2354-3003}%
{}

\authorrunning{C. Awofeso, P. Greaves, O. Lachish, F. Reidl} 

\Copyright{Christine Awofeso, Patrick Greaves, Oded Lachish, Felix Reidl} 



\category{} 

\relatedversion{} 

\nolinenumbers 

\EventEditors{John Q. Open and Joan R. Access}
\EventNoEds{2}
\EventLongTitle{42nd Conference on Very Important Topics (CVIT 2016)}
\EventShortTitle{CVIT 2016}
\EventAcronym{CVIT}
\EventYear{2016}
\EventDate{December 24--27, 2016}
\EventLocation{Little Whinging, United Kingdom}
\EventLogo{}
\SeriesVolume{42}
\ArticleNo{XX}

\usepackage{amsmath,amsfonts,amssymb}

\usepackage{longtable}
\usepackage{url}

\usepackage[dvipsnames]{xcolor}

\usepackage{xspace}

\usepackage{xargs}
\usepackage{xifthen}
\usepackage{framed}

\usepackage{pdflscape}
\usepackage{booktabs}

\usepackage{marginnote} 

\usepackage{thmtools}
\usepackage{thm-restate}

\usepackage{subcaption}
\usepackage{multirow}  
\usepackage{makecell}  

\usepackage{xpunctuate}
\usepackage[all,british]{foreign} 
\redefnotforeign[ie]{i.e\xperiodafter} 
\redefnotforeign[eg]{e.g\xperiodafter}

\usepackage[algo2e, ruled, vlined, linesnumbered, nofillcomment, noend]{algorithm2e}

\SetKwInput{KwInput}{Input}                
\SetKwInput{KwOutput}{Output}              

\SetCommentSty{commentstyle}
\DontPrintSemicolon

\usepackage{tikz}

\usepackage[full,small]{complexity}

\newcommand*\varrule[1][0.4pt]{\leavevmode\leaders\hrule height#1\hfill\kern0pt}
\newcommand{\any}{{\mathord{\color{black!33}\bullet}}}

\newcommand{\GG}{\mathbb G}
\def\adm_#1{ \operatorname{adm}_{#1} }
\def\wcol_#1{ \operatorname{wcol}_{#1} }
\def\scol_#1{ \operatorname{scol}_{#1} }

\newcommand{\orderstrict}{<}
\def\Vleft_#1{L_{#1}}

\newcommand{\pp}{\operatorname{pp}}
\newcommand{\match}{\operatorname{mm}}
\newcommand{\Candidates}{\mathsf{Cand}}

\makeatletter  
\newcommand*{\defeq}{\mathrel{\rlap{%
                     \raisebox{0.3ex}{$\m@th\cdot$}}%
                     \raisebox{-0.3ex}{$\m@th\cdot$}}%
                     =}
\makeatother

\newcommand{\starred}{{}^\star} 

\newcommand{\Neigh}[2]{%
		\ifthenelse{\equal{#2}{}}{N^{#1}}{N_{#1}^{#2}}%
}

\newcommand{\Main}{\texttt{Main}\xspace}
\newcommand{\Update}{\texttt{\ref{alg:Update}}\xspace}
\newcommand{\Augmenting}{\texttt{\ref{alg:Augmenting}}\xspace}

\def\Dvorak{Dvo\v{r}\'{a}k\xspace}

\newcommand*{\circled}[4]{\tikz[baseline=(char.base)]{
             \node[circle,fill=#3,draw=#4,draw,inner sep=1pt,minimum size=1.2em] (char) {\color{#2} #1};}}

\newcommand{\Mark}[1]{\circled{#1}{gray}{white}{gray}}
\makeatletter
\newcommand{\codelabel}[2]{%
	\protected@write \@auxout {}{\string \newlabel {#1}{{#2}{}}}
	\Mark{#2}
}
\newcommand{\coderef}[1]{%
	\Mark{\ref{#1}}\xspace%
}
\makeatother

\renewenvironmentx{leftbar}[2][1=0.5pt, 2=5pt]{%
  \def\FrameCommand{\vrule width #1 \hspace{#2}}\MakeFramed {\advance\hsize-\width \FrameRestore}}%
{\endMakeFramed}

\newcommand{\Marki}{{}^{\scalebox{.55}{\Mark{i}}}}
\newcommand{\One}{{}^{\scalebox{.55}{\Mark{1}}}}
\newcommand{\Two}{{}^{\scalebox{.55}{\Mark{2}}}}
\newcommand{\Three}{{}^{\scalebox{.55}{\Mark{3}}}}
\newcommand{\Four}{{}^{\scalebox{.55}{\Mark{4}}}}

\newcommand{\numnetworks}{214\xspace}

\begin{document}
\maketitle

\begin{abstract}
	The $2$-admissibility of a graph is a promising measure to identify real-world networks which have an algorithmically favourable structure. 
	In contrast to other related measures, like the weak/strong $2$-colouring numbers or the maximum density of graphs that appear as~$1$-subdivisions, the $2$-admissibility can be computed in polynomial time. However, so far these results are theoretical only and no practical implementation to compute the $2$-admissibility exists.

	Here we present an algorithm which decides whether the $2$-admissibility of an input graph~$G$ is at most~$p$ in time~$O(p^4 |V(G)|)$ and space~$O(|E(G)| + p^2)$. The simple structure of the algorithm makes it easy to implement.
  We evaluate our implementation on a corpus of \numnetworks real-world networks
  and find that the algorithm runs efficiently even on networks with millions of edges, that it has a low memory footprint, and that indeed many networks have a small $2$-admissibility. 
\end{abstract}

\section{Introduction}\label{sec:Intro}

Our work here is motivated by efforts to apply algorithms from sparse graph
theory to real-world graph data, in particular algorithms that work
efficiently if certain \emph{sparseness measures} of the input graph are
small. In algorithm theory, specifically from the purview of parametrized algorithms, this approach has been highly successful: by designing algorithms
around sparseness measures like treewidth~\cite{courcelleMonadicSecondorderLogic1990,arnborgEasyProblemsTreedecomposable1991,reedAlgorithmicAspectsTree2003},
maximum degree~\cite{seeseLinearTimeComputable1996}, the size of an excluded minor~\cite{demaineAlgorithmicGraphMinor2005}, or the size of a `shallow' excluded minor~\cite{BndExpII,dvovrak2010deciding}, many hard problems become fixed-parameter tractable, meaning that they can be solved in time~$f(k,\xi) n^{O(1)}$ for some function~$f$, where~$\xi$ is the sparseness measure and~$k$ the problem parameter (\eg solution size). 

For real-world applications, we find a tension regarding sparseness measures: if a measure is too restrictive, it will be large for most graphs in practice. If it is too permissive, then its does not provide enough leverage to design fast algorithms. We aim to identify measures which strike a balance, meaning measures that are small on many real-world networks \emph{and} provide an algorithmic benefit. Additionally, we would like to be able to compute such measures efficiently.

A good starting point here is the \emph{degeneracy} measure, which captures the maximum density of all subgraphs. 
Specifically, a graph is $d$-degenerate if its vertices can be ordered in such a way that every vertex~$x$ has at most~$d$ neighbours that are smaller than~$x$. Such orderings cannot exist for \eg graphs that have a high minimum degree or contain big cliques as subgraphs. For most real-world graphs, however, the degeneracy is indeed small and can be computed quickly. In a survey of 206 networks from various domains by Drange \etal~\cite{drangeComplexityDegenerate2023}, the degeneracy averaged about~23 with a median of~9 (we also replicated this data, see notes on experimental data below). Further,
small degeneracy can improve certain computational tasks significantly. For
example, counting all cliques in a $d$-degenerate graph~$G$ is possible in
time $O(d 3^{d/3} |V(G)|)$~\cite{eppsteinDegenCliques2010} while in general
this problem is $\#W[1]$-hard~\cite{flumCountingHardness2004}\footnote
{This means that counting cliques of size~$k$ is, modulo some
complexity-theoretic assumption, not possible in time $f(k) n^{O(1)}$  for
any function~$f$. }. While we will use subgraph counting as a running example, many other types of problems become tractable on degenerate graphs
(see \eg \cite{philipDegen2012,drangeComplexityDegenerate2023,fomin2021toct,raman2008cocoa}).

However, degeneracy has clear limitations. Assume, for example, that we want to count how often a graph~$H$ appears as a subgraph in another graph~$G$. Let~$H^{(1)}$ and $G^{(1)}$ be the graphs obtained from $H$ and $G$, respectively, by subdividing each edge once. Counting $H^{(1)}$ in~$G^{(1)}$ is then the same as counting $H$ in~$G$, except that $G^{(1)}$ is now $2$-degenerate. Therefore counting~$H^{(1)}$ in $2$-degenerate graphs must be hard\footnote{Here we specifically mean $\#W[1]$-hard, but other hardness results probably translate as well via this simple reduction} if counting $H$ in general graphs is. 
Using stronger complexity-theoretic assumptions, namely the \emph{triangle counting conjecture}\footnote{
There exists a constant~$\epsilon$ such that finding a triangle in a graph cannot be done in time less than~$O(m^{1+\epsilon-o(1)})$.
}, Bera, Pashanasangi, and Seshadhri~\cite{beraCountingSix2020} showed that counting subgraphs with six or more vertices in $d$-degenerate graphs is not possible in time~$o(|E(G)|^{1+\epsilon} f(d))$, where $\epsilon > 0$ stems from the conjecture, thus ruling out linear-time algorithms in this setting.

\smallskip
\noindent
This leaves us with the question whether a sparseness measure exists that
\begin{enumerate}
  \item is small in (many) real-world networks,
  \item allows us to solve problems which remain hard on degenerate graphs,
  \item is computable with few resources, and
  \item can be computed exactly.
\end{enumerate}
We have already provided our motivations for the first two points; the third point derives from our goal to make this measure usable for practical applications on graphs with millions of edges.
The last point warrants some further justification. Our effort of bringing theoretical algorithms into practice means that we need to 
identify domains in which the `sparseness measures approach' could provide more efficient algorithms, since the domains inform what type of problems are of interest. This task is much easier with a sparseness measure whose value is known exactly than with a sparseness measure whose value can only be approximated or computed heuristically.

We propose that the \emph{$2$-admissibility} has the above four properties. As a rough idea, if the degeneracy measures the maximum density of subgraphs in a host graph, the $2$-admissibility measures the maximum density of graphs that appear as $\leq 1$-subdivisions\footnote{A $\leq 1$-subdivision of a graph~$H$ is obtained by subdividing some edges of~$H$ one time.} in the host graph.
We postpone a deeper justification to Section~\ref{sec:measures}, since it involves quite a few technical details, but suffices to say that in contrast to other comparable measures, we know that the \emph{$2$-admissibility} can be computed exactly in polynomial time~\cite{dvorakDomset2013}. Further, we know from previous results~\cite{reidlCounting2023} that graphs like $C_6$, $W_6$, or $K_{5,5}$ can be counted in linear time in graphs with bounded $2$-admissibility, but not in degenerate graphs due to the above mentioned lower bounds by Bera \etal.

\textbf{Theoretical contribution.} 
We present a simple algorithm which decides whether the $2$-admissibility of
an input graph~$G$ is at most~$p$ in time~$O(p^4 |V(G)|)$ using~$O (|E(G)| + p^2)$ space. This improves on a previous algorithm described by \Dvorak~\cite{dvorakDomset2013} with running time $O(|V(G)|^{2p+4})$ and even beats the $2$-approximation with running time $O(|V(G)|^3)$ described in the same paper when~$p < \sqrt{n}$. \Dvorak also provides a theoretical linear-time algorithm for $2$-admissibility in \emph{bounded expansion} classes (which include \eg planar graphs, bounded-degree graphs and classes excluding a minor or topological minor; see Section~\ref{sec:measures} for a definition), however this algorithm relies on a data structure for dynamic first-order model checking~\cite{dvorakFOBndExp13} which certainly is not practical. Our algorithm runs in linear time as long as the $2$-admissibility is a constant, this includes graph classes that do not have bounded expansion.

We designed our algorithm with three goals in mind: The algorithm should have a very low memory footprint (as this is the more strictly constrained resource), it should be easy to implement, and it should be designed `optimistically', meaning that it should only perform work if strictly necessary. We believe that we reached all three goals. On the flip-side, this means that the analysis of the algorithm is more involved. 

\textbf{Implementation and experiments.} 
We implemented our algorithm in Rust and tested it on a collection of~\numnetworks real-world networks, ranging in size from hundreds of edges to 11~million edges. The source code, test data and experimental results can be found at \url{https://github.com/chrisateen/admissibillity-rust}. Experimental results are also attached in the Appendix of this paper.

We ran experiments on a relatively modest machine, proving that our algorithm is indeed very resource efficient. Moreover, our results show that many real-world networks indeed have a small $2$-admissibility.

\section{Preliminaries}\label{sec:prelim}

In this paper all graphs are simple. 
\marginnote{$|G|$, $\|G\|$, $\match(G)$}
For a graph $G$ we use $V(G)$ and $E(G)$ to refer to its vertex set and edge set,
respectively. We use the short-hands $|G| \defeq |V(G)|$ and $\|G\| \defeq |E(G)|$. We call a matching between two sets~$X, Y$ \emph{maximum} if it is the largest possible matching and \emph{maximal} if the matching cannot be increased by adding an edge. We write $\match(G)$ to denote the maximum matching number of~$G$.

\marginnote{$x_1 P x_\ell$, avoids}
For sequences of vertices~$x_1,x_2,\ldots,x_\ell$, in particular paths, we use notations like~$x_1 P x_\ell$,~$x_1 P$ and~$Px_\ell$ to denote the  subsequences $x_1,x_2,\ldots,x_\ell$, $x_1,x_2,\ldots,x_{\ell-1}$ and $x_2,\ldots,x_\ell$, respectively.
Note that further on we sometimes refer to paths as having a
\emph{start-point} and an \emph{end-point}, despite the fact that they are not directed. We do so to simplify the reference to the vertices involved.
A path~$P$ \emph{avoids} a vertex set~$L$ if no interior vertex of~$P$ is contained in~$L$. Note that we allow both endpoints to be in~$L$.

\marginnote{$N_L$, $N_R$, $T^2_L$}
In the following, we will maintain a partition~$L \cup R = V(G)$. Given these two sets, we will write~$N_L(v) \defeq N(v) \cap L$ and~$N_R(v) \defeq N(v) \cap R$ for~$v \in G$. Moreover, we will often need the set of vertices~$T^2_L(v)$ which can be reached from~$v$ via paths of length two that avoid~$L$ and are not in $N(v)$, that is,
$
	T^2_L(v) \defeq \{ x \in V(G)\setminus N(v) \mid N_R(x) \cap N(v) \neq \emptyset \}
$.

\marginnote{$\GG$, $\orderstrict_\GG$}\noindent
In Section~\ref{sec:measures} we will need to define the notion of an \emph{ordered graph}. Given a graph~$G$ and a total ordering~$\orderstrict$ on its vertex set, we define the ordered graph $\GG = (G, \orderstrict)$. We treat $\GG$ as a graph to allow notations like $E(\GG)$, $V(\GG)$ etc.  Given an
ordered graph $\GG$, we write $\orderstrict_\GG$ to denote its ordering relation. We imagine the graph ordered from left (small) to right (large) and will use `left' and `right' to talk about relations between vertices, \eg we say `$u$ lies left of~$v$' to mean $u \orderstrict_\GG v$.

\marginnote{$\min_{\orderstrict_\GG}$, $N^-(\any)$, $\Vleft_\any$}
For a vertex set~$X$ of~$\GG$, we write
$\min_{\orderstrict_\GG} X$ to denote the smallest (leftmost) vertex in~$X$ according to~$\orderstrict_\GG$. For a vertex~$u \in V(\GG)$ we write~$N^-(u) := \{v \in N(u) \mid v \orderstrict_\GG u\}$ for the \emph{left neighbourhood},
\ie all neighbours of~$u$ that came before~$u$ in the ordering. Finally, for~$u \in V(\GG)$ we define the $\emph{prefix set}$ $\Vleft_u$ as the set of all vertices $w\in V(\GG)$ such that $w <_\GG u$.

\section{Admissibility and related measures}\label{sec:measures}

The $2$-admissibility (then just `admissibility') was originally defined in by Kierstead and Trotter~\cite{kiersteadPlanarGraphColoring1994} in their study of planar graph colourings.
The more general notion of \emph{$r$-admissibility} was introduced by \Dvorak~\cite{dvorakDomset2013}. Strictly speaking, $r$-admissibility is a family of measures where $r$ governs how `deep' into the graph we look. We present here the general definition first and then return to the case $r=2$.

\marginnote{$(r,X)$-path packing}
\noindent
An \emph{$(r,X)$-path packing rooted at $v$} is a family of paths $\mathcal P$ with the following properties:
\begin{enumerate}
  \item Every path in $\mathcal P$ has length at most $r$,
  \item all paths have $v$ as start-point and a vertex of $X$ as an end-point,
  \item all paths are vertex-disjoint with the exception of $v$, and
  \item all paths in $\mathcal P$ avoid $X$.
\end{enumerate}

\marginnote{$\pp^r_X$}
\noindent
We write $\pp^r_X(v)$ to denote the maximum size of a $(r,X)$-path packing rooted at $v$.

\marginnote{\small$r$-admissi\-bility}\noindent
The $r$-admissibility $\adm_r(\GG)$ of an ordered graph $\GG$ is now defined as\footnote{Note that some authors define the admissibility as $\max_{v \in V(\GG)} \pp_{\Vleft_v}^r(v)+1$ to align with other measures.
}
\vspace*{-5pt}
\begin{align*}
	\adm_r(\GG) &:= \max_{u \in V(\GG)} \pp_{\Vleft_u}^r(u),
\end{align*}
recall that~$L_u$ is the prefix set of all vertices left of~$u$ in~$\GG$.
The $r$-admissibility $\adm_r(G)$ of a graph $G$ is the minimum of $\adm_r(\GG)$ taken over all orderings $\GG$ of~$G$.

\marginnote{strong and weak colouring}
Two closely related measures are the \emph{strong $r$-colouring number} $\scol_r$ and the \emph{weak $r$-colouring number} $\wcol_r$. Both are based on measuring the size of certain `reachable' sets from every vertex in an ordered
graph~$\GG$. Specifically, we definde the \emph{weakly $r$-reachable} set $W_\GG^r(u)$ and the \emph{strongly $r$-reachable set} $S_\GG^r(u)$ as
\begin{align*}
  S^r_\GG(u) &= \{ v \orderstrict_\GG u \mid \exists~\text{$u$-$v$-path~$P$ of length at most~$r$ such that}~ u = \min_{\orderstrict_\GG} (P\setminus v) \} \\
  W^r_\GG(u) &= \{ v \orderstrict_\GG u \mid \exists~\text{$u$-$v$-path~$P$ of length at most~$r$ such that}~ v = \min_{\orderstrict_\GG} P \}   
\end{align*}
In words, the vertices~$v$ that are strongly $r$-reachable from~$u$ lie to the left of~$u$ and are connected to~$u$ via a path of length~$\leq r$ that only uses vertices larger than~$u$, with the exception of~$v$ itself. The vertices
that are weakly $r$-reachable from~$u$ also lie to the left of~$u$, but the paths conneting them to~$u$ are now allowed to use vertices right of~$v$. 

The example below should help tease apart these definitions. Depicted are the left neighbourhood~$N^-(x)$ of~$x$, a maximum $2$-path packing rooted at~$x$, the set of strongly $2$-reachable vertices~$S^2(x)$, and the set of weakly
$2$-reachable vertices~$W^2(x)$.
\begin{center}
  \includegraphics[width=0.8\textwidth]{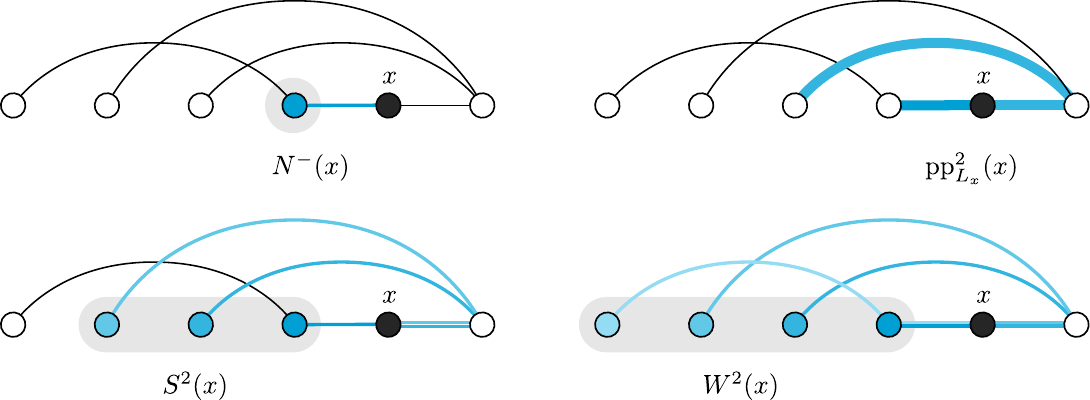}
\end{center}

\vspace*{-5pt}
\noindent
The two measure~$\scol_r$ and~$\wcol_r$ of an ordered graph~$\GG$ are now defined, respectively, as 
\[
  \scol_r(\mathbb G) := \max_{u \in V(\GG)} |S^r(u)| \quad \text{and} \quad
  \wcol_r(\mathbb G) := \max_{u \in V(\GG)} |W^r(u)|.
\]
For an unordered graph~$G$, the colouring numbers $\scol_r(G)$ and $\wcol_r(G)$ are again defined as the minimum over all possible orderings~$\mathbb G$ of~$G$.

All three measures are generalizations of the degeneracy. Specifically, for~$r=1$ we have that $S^1(u) = W^1(u) = N^-(u)$ and~$\pp^1_{L_u} = |N^-(u)|$, therefore $\adm_1(\mathbb G) = \scol_1(\mathbb G) = \wcol_1(\mathbb G)$ and all three measures equal the degeneracy.
For~$r \geq 2$, the measures are usually different. However, they
are closely related in the sense that for any ordered graph~$\mathbb G$ and any~$r \in \mathbb N$, it holds that 
\[
    \adm_r(\mathbb G) \leq \scol_r(\mathbb G) \leq \wcol_r(\mathbb G) 
    \leq (r^2 \adm_r(\mathbb G))^r,
\],
where the first two inequalities follow quite straightforward from the definitions and the last inequality is due to \Dvorak~\cite{dvorakWeightedSublinearSeparators2022} (improving on an earlier bound
by the same author~\cite{dvorakDomset2013}). Note that the above inequalities also hold if we replaced the ordered graph with an unordered graph.



\marginpar{\small bounded expansion}
For theoretical considerations, the three measures are often exchangable since if one is bounded, the other two are as well. The weak colouring number provides the most `structure' to work with and has been used to great success in the study of \emph{bounded expansion} classes (see \eg \cite{drangeDomsetKernel2016,amiriDistDomset2018,groheCoveringNowhereDense2018,reidlCharBoundExp2019}) as well as for designing algorithms~\cite{reidlCounting2023,dvorakDomset2013}. The original definition of bounded expansion classes involves the notion of shallow minors~\cite{Sparsity} which we will forgo here, since Zhu~\cite{zhuGeneralizedCols2009} showed that bounded expansion classes are exactly those classes in which the weak $r$-colouring number is bounded for \emph{any}~$r$, \ie there exists some function~$f$ such that for every member~$G$ of the class it holds that~$\wcol_r(G) \leq f(r)$. Note that we can equivalently use~$\adm_r$ or~$\scol_r$ in this definition since the measures bound each other.

However, both the weak and strong $2$-colouring numbers are \NP-hard to compute exactly~\cite{colHardness2023} and for that reason, studies on practical values of these measures~\cite{WcolExperimental,TurbochargeWcol} are somewhat inconclusive as large values on some networks might simply be a failure of the respective heuristic\footnote{Some approximation algorithms exist, but the ratios are not good enough to paint a clear picture}.
The same holds true for comparable measures like the maximum density of
graphs that appear as $1$-subdivisions~\cite{muziBeingEvenSlightly2017} which also bounds and is bounded by the $2$-admissibility. So for the purpose of finding an \emph{exact} sparseness measure that can be computed efficiently, only the $2$-admissibility is a suitable candidate.

We close this section by showing a few basic properties of admissibility and path-packings.

\begin{restatable}{lemma}{lemmaNumEdges}\label{lemma:NumberOfEdges}
	For every graph $G$ and integers $p, r \geq 1$,	if $\|G\| > p|G|$ then $\adm_r(G) > p$.
\end{restatable}
\begin{proof}
If $\|G\| > p|G|$, then for every order $\orderstrict$ of the graph $G$, there exists a vertex $v$ such that $|\Vleft_v| > p$ and hence $\pp^r_{\Vleft_v}(v)>p$.
\end{proof}

\begin{restatable}{lemma}{lemmaSmallMatching}\label{lemma:small-pp-small-matching}
	Let $G$ be a graph, $L\subseteq V(G)$, $v \in L$, and $\M_v$ be a matching between the vertices of $N_R(v)$ and of $T^2_L(v)$. 
	Then, $\pp^2_L(v) \geq |N_L(v)| + |M_v|$ and equality holds if $\M_v$ is a maximum matching. 
\end{restatable}
\begin{proof}
  Fix $L$ to be an arbitrary subset of $V(G)$ and $v$ an arbitrary vertex in $L$.
  Let $R = V(G)\setminus L$ and  $\M_v$ a matching between the vertices of $N_R(v)$ and of $T^2_L(v)$.
  
  We show first that  $\pp^2_L(v) \geq |N_L(v)| + |M_v|$.
  Afterwards we show that, if $M_v$ is a maximum matching, then $\pp^2_L(v) \leq |N_L(v)| + |M_v|$. 
  The lemma follows. 
  
  Let $\hat{\ell}$ be the number of edges of $\M_v$ and~$\{x_iy_i\}_{i \in [\hat{\ell}]}$ be the edges of $\M_v$ with~$x_i \in N_R(v)$ and~$y_i \in T^2_L(v)$. Construct the paths~$\mathcal P' \defeq \{ vu \mid u \in N_L(v) \} \cup \{vx_iy_i \mid i \in [\hat{\ell}]\}$. Since the sets $N_R(v)$, $N_L(v)$ and $T^2_L(v)$ are all pairwise disjoint and the paths~$vx_iy_i$ only intersect in~$v$, we conclude that~$\mathcal P'$ is a $(2,L)$-path packing of size $|M_v| + |N_L(v)|$. Thus, by definition, $\pp^2_L(v) \geq |M_v| + |N_L(v)|$.
  
  Suppose that $\M_v$ is a maximum matching.
  Let $\mathcal P$ be a $(2,L)$-path packing~ rooted at~$v$ of size $\pp^2_L(v)$.
  Let $\ell$ be the number of paths of $\mathcal P$ that have length one. Note that the remaining~$s \defeq |\mathcal P| - \ell$
  paths all have length two and take the form~$vx_iy_i$, $1 \leq i \leq s$, where~$x_i \in N_R(v)$ and~$y_i \in T^2_L(v)$. 
  By definition, all the paths of length one of $\mathcal P$ have an end point in $N_L(v)$ and hence $s \leq |N_L(v)|$.
  Also by definition, all the $x_i$ and $y_i$ are distinct, therefore the edges~$\{x_iy_i\}_{i \in [s]}$ are a matching between $N_R(v)$  and of $T^2_L(v)$. Thus, $s$ is at most the size of a maximum matching in $G$ between the vertices of $N_R(v)$ and of $T^2_L(v)$.
  Consequently, as $\pp^2_L(v) = s + \ell$, it holds that $\pp^2_L(v) \leq |N_L(v)| + |M_v|$.
\end{proof}

\begin{restatable}{lemma}{lemmaCandStaysCand}\label{lemma:Cand-stays-Cand}
		Let $G$ be a graph.
		For every $L' \subseteq V(G)$, $L\subseteq L'$ and $v \in L'$ it holds that $\pp^2_L(v) \leq \pp^2_{L'}(v)$. 
\end{restatable}
\begin{proof}
  Let $L',L$ and $v$ be selected arbitrarily so that $L'\in V(G)$, $L\subset L'$ and $v \in L'$.
  Let $R = V(G)\setminus L$ and $R' = V(G)\setminus L'$.  
  
  Let $\mathcal P$ be a $(2,L)$-path packing rooted at~$v$ of size $\pp^2_L(v)$.
  By definition, $\mathcal P$ contains all the vertices in $N_L(v)$ and a maximum matching $\M$ of size $s$ between the vertices of $N_R(v)$ and of $T^2_L(v)$. 
  Let~$\{x_iy_i\}_{i \in [s]}$ be the edges of $\M$ with~$x_i \in N_R(v)$ and~$y_i \in T^2_L(v)$.
  
  Let $\M'$ be a matching that contains exactly every edge $x_iy_i$ of $\M$ such that $x_i \in N_{R'}(v)$ and $y_i \in N_{L'}^{2}(v)$.  
  We notice that, as $L \subseteq L'$, every vertex $y_i$ in an edge of $\M$ is also in $L'$.
  So, the only reason an edge $x_iy_i$ of $\M$ will not be in $\M'$ is if $x_i \in L'$ (and not in $R'$) which in turn means it is in $N_{L'}(v)$.
  Thus,  $N_{L'}(v)$ contains all the vertices $x_i$ such that $x_iy_i$ is an edge in $\M$ but not in $\M'$.
  $N_{L'}(v)$ also contains all the vertices of $N_L(v)$, since $L\subseteq L'$. 
  Consequently, (*) $|N_{L'}(v)| + \|\M'\| \geq |N_L(v)| + \|\M\|$.
  By Lemma~\ref{lemma:small-pp-small-matching}, $\pp^2_{L'}(v) \geq |N_{L'}(v)| + \|\M'\|$ and $\pp^2_L(v) = |N_L(v)| + \|\M\|$ and therefore, by (*), $\pp^2_{L'}(v) \geq \pp^2_{L}(v)$.
\end{proof}

\noindent

\section{The theoretical algorithm}
\marginnote{\Main}
We present here the Main Algorithm (\Main).
The input to the algorithm is a graph $G$ and a parameter $p$.
We assume that $\|G\| \leq pn$, since this can easily be checked and if it does not hold, then by Lemma~\ref{lemma:NumberOfEdges}, the $2$-admissibility number of $G$ is strictly larger than $p$. 

Let us for now assume that we have access to an Oracle that, given a subset $L$ of $V(G)$, can provide us with a vertex~$v \in L$
such that~$\pp^2_L(v) \leq p$ if such a vertex exists. With the help of this Oracle, the following greedy algorithm returns an ordering $\orderstrict$ such that $\adm_2(G,\orderstrict)\leq p$ if one exists and otherwise returns FALSE:

\begin{enumerate}
	\item Initialise $L \defeq V(G)$ and~$i = |G|$, as well as the Oracle.
	\item Find~$v \in L$ such that~$\pp^2_L(v) \leq p$ using the Oracle. If no such vertex exists return FALSE.
	\item Remove $v$ from~$L$, set $v_i$ to be $v$ and decrease $i$ by one. Go to the previous step unless~$L = \emptyset$.
	\item Return the order~$v_1,\ldots, v_{|G|}$.
\end{enumerate}

\noindent
The justification for why this greedy strategy works is a well-known property of $r$-admissibility, captured by the following lemma.

\begin{restatable}{lemma}{lemmaRadm}\label{lemma:r-admissibility}
	A graph $G$ has $r$-admissibility $p$, if and only if, for every non-empty $U
	\subseteq V$, there exists a vertex $u \in U$ such that $\pp^r_{V\setminus U}(u) \leq
	p$.
\end{restatable}
\begin{proof}
  Suppose first that for every non-empty $U
  \subseteq V$, there exists a vertex $u \in U$ such that $\pp^r_{V\setminus U}(u) \leq
  p$.
  Then, an $r$-admissible order of $G$ can be found as follows, first 
  initialise the set $U$ to be equal to $V$ and $i$ to $|G|$, and then repeat the following two steps until $U$ is empty: 
  (1) find a vertex $u \in U$ such that $\pp^r_{V\setminus U}(u) \leq
   p$ removing it from $U$ and adding it in the $i$'th place of the order
  (2) decrease $i$ by $1$.
   
  We note that by construction, the $r$-admissibility of the order we got is at most $p$.
  Thus, $G$ has $r$-admissibility $p$.

  Suppose that $G$ has $r$-admissibility $p$. 
  Then, there exists an ordering, $\orderstrict$ of $V(G)$ such that $\adm_r(G, \orderstrict) \leq p$.
  Let $u_1, u_2, \ldots u_{|G|}$ be the vertices of $G$ and suppose
   that for every $i,j\in [|G|]$, such that $i<j$, it holds that $u_i \orderstrict u_j$.
  
  Let $u_k\in U$ be the vertex with the maximum index in $U$.
  We note that $U\subseteq \Vleft_{u_k}$.
  
  Assume for the sake of contradiction that $\pp^r_U(u_k) > p$. 
  Then, there exists an $(r,U)$-path packing $Pack$ rooted at $u_k$ of size $p+1$. 
  By definition, every path in $Pack$ starts in $u_k$ and ends in a vertex in
  $U$.

  Thus, since $U \subseteq \Vleft_{u_k}$, every path $P$ in $Pack$ 
  either avoids $\Vleft_{u_k}$, or includes a vertex from $\Vleft_{u_k}$.
  If the second case holds for such a path $P$, then it has sub-path that
  starts in $u_k$ and ends in a vertex in $\Vleft_{u_k}$ and avoids $\Vleft_{u_k}$. 
  Now, we can construct a new $(r,\Vleft_{u_k})$-path packing rooted at $u_k$ with the same size as $Pack$ ($p+1$), by iterating over all paths $P$ and taking them if the avoid $\Vleft_{u_k}$ and otherwise taking their subpath that avoid $\Vleft_{u_k}$.
  since $(G, \orderstrict)$ has $r$-admissibility at most $p$.
\end{proof}

\noindent

\subsection{The Oracle}

In the following, the set $L$ is the set of all the vertices that the Oracle has yet to return and the set $R$ is equal to $V(G)\setminus L$.

We note that, beyond providing the Oracle with the graph $G$ and parameter $p$, there is no need to update the Oracle since it has the knowledge of the vertex it returned.  To answer the queries efficiently, the Oracle maintains additional information for each vertex which enables it to answer queries in an amortized running time that is polynomial in~$p$.
Specifically, the Oracle maintains a set of vertices $\Candidates$ (short for `candidates'), and the following data structures for each vertex~$v \in L$:

\newcommand{\LMV}[1]{L^{\M}_{#1}}
\newcommand{\RMV}[1]{R^{\M}_{#1}}
\newcommand{\LMVP}[1]{L^{\M'}_{#1}}
\newcommand{\RMVP}[1]{R^{\M'}_{#1}}
\newcommand{\LMVS}[1]{L^{\M^*}_{#1}}
\newcommand{\RMVS}[1]{R^{\M^*}_{#1}}
\newcommand{\ME}[1]{E^{\M}_{#1}}

\medskip
\noindent\begin{tabular}{ll}
	$N_L(v)$, $N_R(v)$ & The sets $N_L(v)$ and $N_R(v)$ as defined above\\
	$\M_v$          	& A matching $(\LMV{v},\RMV{v},\ME{v})$ 
	                     with $\RMV{v} \subseteq N_R(v)$, 
	$\LMV{v}\subseteq T^2_L(v)$, and\\ 
	&	matching edges $\ME{v}$ \\
\end{tabular}

\medskip
\noindent
While technically~$N_L(v)$ and $N_R(v)$ are data structures distinct from the neighbourhood they represent, the Oracle will always update these data structures first to ensure consistency. For that reason, we avoid introducing additional notation and will simply equivocate the data structure and the neighbourhood it represents. 
The Oracle has three types of operations: \\

\noindent\textbf{Data structures initialisation.}
Given $G$ and $p$ the Oracle initialises $\Candidates$ to be the set of all vertices of degree at most $p$ in $G$ and, for every $v\in V(G)$, it sets $N_L(v) = N_G(v)$, $N_R(v) = \emptyset$, and $\M(v)$ to be the empty graph.

\smallskip

\noindent\textbf{Return vertex.} 
The Oracle returns a vertex from the set $\Candidates$ if it is not empty and otherwise it returns FALSE.

\smallskip

\noindent\textbf{Update data structures.}
The goal of the update is to ensure that the following conditions hold after a vertex is returned (removed from $L$). The conditions will be enforced by the algorithm~\ref{alg:Update} which runs just before the next vertex is returned:

\begin{enumerate}[(1)]
	\item $\Candidates$ contains all vertices $u\in L$ with $\pp^2_L(u) \leq p$,
	\item for every $w \in V(G)$, the data structures representing $N_L(w)$ and $N_R(w)$ contain the vertices they should by the definition of these sets, and
	\item for every $w \in L$, the matching $\M_w$ is a maximal matching between the vertices in $N_R(w)$ and the vertices in $T^2_L(w)$,
	and if $w\not\in \Candidates$ then also $|N_L(w)| + |\M_w| \geq p + 1$
\end{enumerate}

\noindent
The update will be perfomed by the Algorithm \Update, which in turn calls the Algorithm \Augmenting presented later. For the time being, it is sufficient to know that the following lemma (which we prove after presenting \Augmenting) holds:
\begin{restatable}{lemma}{lemmaRada}\label{lemma:Augmenting}
	Given a vertex $v$ as input,	
	\Augmenting returns FALSE if $\M_v$ is a maximum matching between~$N_R(v)$ and~$T^2_L(v)$. Otherwise, it returns a new maxi\emph{mal} matching between these sets with~$|\M_v|+1$ edges.
\end{restatable}

\begin{algorithm2e}[!hbt]
	\SetAlgoRefName{Update}
  \SetKw{Break}{break}
  \SetKw{Continue}{continue}
	\DontPrintSemicolon
	\KwIn{A vertex $v$ that was in
		 $\Candidates$}
	$\mathsf{Check} = \emptyset$\;
	\tcp{\codelabel{update:neighs}{1}Update neighbours in $N(v)$} 
	\For{$u\in N(v)$}{ 
		Remove~$v$ from~$N_L(u)$\;
		Add~$v$ to~$N_R(u)$\;
	}	
	\tcp{\codelabel{update:matching}{2}Ensure maximality of affected matchings}
	\For{$u\in N_L(v)$}{
		\For{$w\in N_L(v)$}{
			\If {$w \not\in \LMV{u} \cup N_L(u) \cup \{u\}$}{
				Add the edge $vw$ to $\M_u$\;
			 \Break\; 
			}
		}
		Add~$u$ to $\mathsf{Check}$\;
	}
	\tcp{\codelabel{update:neighs2}{3}Update relevant neighbours in $T^2_L(v)$}
	\For{$u\in \bigcup_{x \in \RMV{v}} N_L(x) \setminus \{v\}$}{
		\If {$v \not \in \LMV{u}$}{
			\Continue\;
		}
		Let $vz$ be the edge in $\M_u$ incident to $v$\;
		Remove $vz$ from $\M_u$\;
		\For{$y \in N_L(z)$}{
			\If{$y\not\in N_L(u) \cup \LMV{u} \cup \{u\}$}{
				Add the edge $yz$ to $\M_u$\;
			 \Break\;
			}
		}
		Add~$u$ to $\mathsf{Check}$\;
	}
	\tcp{\codelabel{update:check}{4}Check whether affected vertices become candidates}
	\For{$u \in \mathsf{Check}$}{
		\If{$|\M_u| + |N_L(u)| = p$ and $u\not\in \Candidates$}{
			\If{ \textnormal{\texttt{Augmenting}(u)} = FALSE}{
				Add $u$ to $\Candidates$\;
			}
		}
	}
	$\M_v := \emptyset$\;
	
	%

	\caption{\label{alg:Update}%
  \footnotesize
	Updates the Oracle data structures when moving a vertex~$v$ from~$L$ to~$R$.}
\end{algorithm2e}

\noindent
Let us now prove that the three conditions over the Oracle data structures hold after they are initialised and that they are maintained by \Update. This, by Lemma~\ref{lemma:r-admissibility}, already implies that \Main works correctly.

Afterwards we will bound the running time of \Update and the number of times that \Update executes  \Augmenting with the same input. Both lemmas are required in order to bound the running time of \Main. After that, we prove Lemma~\ref{lemma:Augmenting} and bound the overall running time of \Augmenting to then present the main theorem of this section.

\begin{lemma}\label{lemma:Update}
	The first three data structure conditions hold after the Oracle data structure initialisation and before and after every time \Update is executed.
\end{lemma}
\begin{proof}
	Prior to the first time \Update is executed we have that $L = V(G)$. At this stage, for every $v\in V(G)$ the matching $\M_v$ must be empty and~$\pp^2_L(v)$ is equal to the degree of~$v$.
	Therefore $\Candidates$ must contain every vertex that has degree at most~$p$. This is precisely how we initialise the relevant data structures and we conclude that conditions (1)-(3) are met.

	Let us now deal with maintaining Condition~(2). The loop~\coderef{update:neighs} in \Update updates the data structures 
	$N_R$, $N_L$ for all vertices in~$N(v)$, where~$v$ is the vertex that was just moved from~$L$ to~$R$. It is easy to see that these are the only vertices for which the values of~$N_R$ and~$N_L$ change, hence we conclude that Condition~(2) is maintained by \Update. Moreover, we can assume that it holds when the parts~\Mark{2}--\Mark{4} of \Update are executed.

	Let us now deal with the Conditions~(1) and~(3). In order to talk about the state of various data structures in the algorithm, we will use left superscript annotations~$\Marki$, $i \in \{1,2,3,4\}$, to denote the state of a data structure \emph{after} the execution of Loop~\Mark{i}. Sets and data structures without annotation refer to their state before the execution of \Update. To avoid very convoluted
	notation, we write \eg $\Marki N_L$ instead of~$\Marki N_{\Marki L}$.
	Note that with this notation, $v \in L$ and~$v \not\in R$ but~$v \not \in {\One L}$ and~$v \in {\One R}$.	Also note that the sets~$\Marki L$ and $\Marki R$ remain the same during \Update. 

	Before we proceed, observe that all~$\Marki \M_v$ are the same: Loop~\Mark{1} does not change matchings and Loop~\Mark{2} does not change~$\M_v$, because~$v \not \in N_L(v)$. Loop~\Mark{3} enforces that~$u \neq v$, therefore the matching~$\M_v$ is not changed. For these reasons, $v$ is not added to $\mathsf{Check}$ and therefore Loop~\Mark{4} does not modify~$\M_v$.
		
	Let us now show that that Condition~(3) is upheld. To that end, let us first establish which matchings should \emph{not} be changed by \Update:

	\begin{claim}
		Let~$u \in L \setminus \{v\}$ and suppose~$v \not \in V(\M_u) \cup N_L(u)$. Then~$\M_u$ is a maximal matching between~$\Four N_R(u)$ and~$\Four T^2_L(u)$.
	\end{claim}
	\begin{proof}
		First note that since~$v \not \in V(\M_u)$ all edges of~$\M_u$ are indeed between~$\Four N_R(u)$ and~$\Four T^2_L(u)$.

		Assume now towards a contradiction that the edge~$ab$, $a \in \Four N_R(u)$
		and~$b \in \Four T^2_L(u)$, can be added to~$\M_u$ to obtain a larger matching. Since~$\M_u$ was, by Condition~(3), a maximal matching between $N_R(u)$ and~$T^2_L(u)$ it follows that either~$a \not \in N_R(u)$ or~$b \not \in T^2_L(u)$. In the first case, we conclude that~$a = v$ since it is the only vertex that changed from~$L$ to~$\Four R$. But then~$v \in N_L(u)$, contradicting the conditions of the claim. Consider therefore the case that
		$a \neq v$, so in particular we have that $a \in N_R(u)$, and $b \not \in T^2_L(u)$. Since vertices cannot move from~$R$ to~$L$ and~$b \not \in N(u)$, the only reason for~$b \not \in T^2_L(u)$ is that the only $2$-path from~$u$ to~$b$ is ~$uvb$. However, our construction also shows the existence of the path~$uab$ with~$a \neq v$, a contradiction.
	\end{proof}

	\noindent
	As a corollary to this claim, we find that a matching~$\M_u$ only needs to be updated if either~$v \in N_L(u)$ or~$v \in V(\M_u)$. Since~$\M_u$ does not contain any vertices from~$N_L(u)$, these two cases are mutually exclusive.

	Let us first verify that for any~$u \in L$ such that~$v \not \in V(\M_u) \cup N_L(u)$ it indeed holds that~$\Four \M_u = \M_u$: Loop~\Mark{1} does not modify any matchings. Loop~\Mark{2} does not modify $\M_u$ since by assumption $v \not \in N_L(u)$ and thus~$u \not \in N_L(v)$, in particular $u$ will not be added to $\mathsf{Check}$. Loop~\Mark{3} will not modify $\M_u$ since~$v \not \in V(\M_u)$, so if the loop considers~$u$, then it continues in line~13 without modifying~$\M_u$ and $u$ is not added to $\mathsf{Check}$. Since no loop so far added~$u$ to $\mathsf{Check}$, we can conclude that Loop~\Mark{4} never considers~$u$ and therefore does not modify~$\M_u$. We conclude that $\Four \M_u = \M_u$, as desired.

	Consider now a vertex~$u \in L$ with~$v \in N_L(u)$.
	Let us first verify that~$\Two \M_u$ is a maximal matching between~$\Two N_R(u)$ and~$\Two T^2_L(u)$. Loop~\Mark{1} does not modify the matching. Since~$v \in N_L(u)$, Loop~\Mark{2} is entered. If there exists an edge~$vw$ with $w \in \One T^2_L(u)$ and~$w \not \in V(\One \M_u)$, then \Update must add $vw$ to~$\One \M_u$ to preserve maximality. Lines~6 and~7 locate such a vertex~$w$ if it exists since necessarily~$w \in N_L(v) = \One N_L(v)$. As no further edge incident to~$v$ could then be added, \Update can forgo further checks and breaks the inner loop in line~9. By maximality of~$\M_u$, no further edge could possibly be added to the matching and we conclude that~$\Two \M_u$ is indeed maximal between~$\Two N_R(u)$ and~$\Two T^2_L(u)$ and therefore also maximal between~$\Three N_R(u)$ and~$\Three T^2_L(u)$.
 
	Consider now a vertex~$u \in L$ with~$v \in \M_u$. Since~$v \in L$, this is only possible if~$u \in T^2_L(v)$. Since the Oracle does not store this set, we need to justify line~11 by showing that $u \in \bigcup_{x \in R^\M_v} N_L(x)$.

	\begin{claim}
		Assume~$v \in V(\M_u)$. Then there exists a vertex~$x \in R^\M_v$ such that
		$u \in N_L(x)$.
	\end{claim}
	\begin{proof}
		Since~$v \in L$ and $v \in V(\M_u)$ we conclude that~$v \in T^2_L(u)$. This of course also means that~$u \in T^2_L(v)$, so clearly there exist at least one vertex~$x \in N_R(v)$ with~$u \in N_L(x)$. If all such vertices were not contained in~$x \in R^\M_v$, then clearly also~$u \not \in L^\M_v$ and we could add the edge~$ux$ to~$\M_v$, contradicting maximality.
	\end{proof}

	\noindent
	We conclude that Loop~\Mark{3} will indeed consider~$u$ and, by assumption, reach line~14 in this iteration. Note that Loop~\Mark{2} did not modify~$\One \M_u$ and we therefore have that~$\Two \M_u = \One \M_u = \M_u$.
	Let~$vz$ be the edge incident to~$v$ in~$\M_u$. By maximality and since~$v \in \M_u$, no edge between $\Two N_R(u)$ and~$\Two T^2_L(u)$ can be added to~$\M_u$. Therefore, the only way in which~$\M_u' \defeq \M_u \setminus \{vz\}$ is \emph{not} maximal in $\Two N_R(u) \cup \Two T^2_L(u)$ is if there exists some vertex~$y$ so that~$yz$ can be added to~$\M_u'$. Since~$z \in R$, $y$ must be in~$L$ and therefore~$y \in N_L(z)$. The edge~$yz$ can be added to~$\M_u'$ if~$y \not \in L \cap \M_u'$ and $y \not \in N_L(u) \cup \{u\}$. Note that~$L \cap \M_u' = L \cap \Two \M_u = L \cap \M_u$, therefore lines~16 and~17 select~$y$ correctly and add such an edge~$yz$ to~$\M'_u$ if possible. If such an edge is added, the vertex~$z$ is now matched again and by our previous observation $\Three \M_u = \M_u' \cup \{yz\}$ must be maximal. Otherwise, $\Three \M_u = \M_u$ is proven to be maximal between~$\Two N_R(u) = \Three N_R(u)$ and~$\Two T^2_L(u) = \Three T^2_L(u)$.

	We have shown that all matchings~$\M_u$, $u \in L$, are maximal when we reach Loop~\Mark{4}. By Lemma~\ref{lemma:Augmenting} the algorithm \Augmenting preserves maximality and therefore all matchings~$\M_u$ for~$u \in \Four L$ are maximal between~$\Four N_R(u)$ and~$\Four T^2_L(u)$. Let us now address the second part of Condition~(3), namely that if~$u \not \in \Four \Candidates$, then we must prove that~$|\Four N_L(u)| + |\Four \M_u| \geq p + 1$.

	Consider such a vertex. Since~$u \not \in \Four \Candidates$ we have that~$u \not \in \Candidates$, hence~$|N_L(u)| + |\M_u| \geq p + 1$ by Condition~(3). If~$u \not \in \Three\mathsf{Check}$, then~$\Three N_L(u) = N_L(u)$
	and~$\Three \M_u = \M_u$ and hence
  \vspace*{-5pt}
	\[
		|\Four N_L(u)| + |\Four \M_u| = |\Three N_L(u)| + |\Three \M_u| \geq p + 1.
	\]
	Assume now that~$u \in \Three\mathsf{Check}$, which means that line~23 is executed for~$u$ and accordingly~$|\Three N_L(u)| + |\Three \M_u| = p$.
	In order to justify this line, we need the following observation:

	\begin{claim}\label{claim:down-by-one}
		$|\Three N_L(u)| + |\Three \M_u| \geq |N_L(u)| + |\M_u| - 1 $
	\end{claim}
	\begin{proof}
		Recall that~$\One N_L(u) = \Three N_L(u)$. First consider the case that~$\Three N_L(u) \neq N_L(u)$, which is only the case if Loop~\Mark{1} removed~$v$ from~$N_L(u)$ to construct~$\One N_L(u)$ and hence~$|\Three N_L(u)| = |N_L(u)|-1$. As we have established above, Loop~\Mark{3} is not executed for~$u$ and therefore~$\Three \M_u = \Two \M_u$. Since Loop~$\Mark{2}$ adds at most one edge to~$\One \M_u$, we have that~$|\Two \M_u| \geq |\One \M_u| = |\M_u|$. We conclude that in this case the claim holds.

		Consider now the case that~$\Three N_L(u) = N_L(u)$. Then~$\Two \M_u = \M_u$
		since Loops~\Mark{1} and~\Mark{2} are not executed for~$u$. Loop~\Mark{3} constructs $\Three \M_u$ from~$\Two \M_u$ by removing at most one edge (and potentially adding one), therefore $|\Three \M_u| \geq |\M_u| - 1$ and the claim holds.
	\end{proof}

	\noindent
	Since~$|N_L(u)| + |\M_u| \geq p + 1$, we have that $|\Three N_L(u)| + |\Three \M_u| \geq p$. First consider the case~$|\Three N_L(u)| + |\Three \M_u| \geq p + 1$.
	Then lines~23 and~24 are not executed and therefore~$u \not \in \Four \Candidates$
	with~$|\Four N_L(u)| + |\Four \M_u| = |\Three N_L(u)| + |\Three \M_u| \geq p + 1$. Otherwise, $|\Three N_L(u)| + |\Three \M_u| = p$ and line~23 is executed. Since we assumed that~$u \not \in \Four \Candidates$, line~24 was not executed, therefore \Augmenting(u) returned TRUE. By Lemma~\ref{lemma:Augmenting}, this means that~$|\Four \M_u| = |\Three \M_u| + 1$
	and we conclude that~$|\Four N_L(u)| + |\Four \M_u| \geq p + 1$.
	We conclude that Condition~(3) holds after each call to \Update.

	Finally let us prove that \Update maintains Condition~(1), \ie we show that each vertex	$u \in \Four L$ with~$\pp^2_{\Four L}(u) \leq p$ is contained in $\Four \Candidates$. Since~$v \not \in \Four L$ we can assume~$u \neq v$ in the sequel. Note that~$\Four L \subset L$, so by applying Lemma~\ref{lemma:Cand-stays-Cand} we have that~$\pp^2_{\Four L}(u) \leq \pp^2_L(u)$.
	Hence if~$\pp^2_{L}(u) \leq p$, then~$u \in \Candidates$ and therefore
	$u \in \Four \Candidates$. Consider therefore a vertex~$u$ with~$\pp^2_L(u) > p$ and~$\pp^2_{\Four L}(u) \leq p$, in particular this means that $u \not \in \Candidates$. By Condition~(3) then $|N_L(u)| + |\M_u| \geq p + 1$. By Lemma~\ref{lemma:small-pp-small-matching}, $|\Three N_L(u)| + |\Three \M_u| \leq \pp^2_{\Three L}(u) = \pp^2_{\Four L} \leq p$.  Therefore~$N_L(u) \neq \Three N_L(u)$ or~$\M_u \neq \Three \M_u$ and either Loop~\Mark{2} or Loop~\Mark{3} added~$u$ to~$\mathsf{Check}$, therefore \Update will execute line~22 for~$u$. By Claim~\ref{claim:down-by-one}, we have that~$|\Three N_L(u)| + |\Three \M_u|
	\geq |N_L(u)| + |\M_u| - 1 = p$ and we can conclude that $|\Three N_L(u)| + |\Three \M_u| = p$. Therefore, the conditions on line~22 are met and line~23 is executed for~$u$. The call to \Augmenting(u) must return FALSE as otherwise $|\Four \M_u| = |\Three \M_u| + 1$ which would lead to~$\pp^2_{\Four L}(u) = p +1$, contradicting our assumptions about~$u$. Consequently, line~24 is executed for~$u$ and we conclude that~$u \in \Four \Candidates$. In summary, all vertices $u \in \Four L$ with~$\pp^2_{\Four L}(u) \leq p$ are indeed contained in $\Four \Candidates$ and Condition~(1) holds after each call to \Update.
\end{proof}

\newcommand{\out}{\mathsf{out}}
\begin{algorithm2e}[!htb]
	\SetAlgoRefName{Augmenting}
	\DontPrintSemicolon
	\KwIn{A vertex $v$}
	\KwOut{\textbf{TRUE} if the algorithm added one more edge to $\M_v$, \textbf{FALSE} otherwise.}
	\tcp{Construct auxiliary matching graph}
	Create an empty graph $\vec H_v$ with vertices $V(\M_v)$\;
	$S:= \emptyset$, $T:= \emptyset$, $\out \defeq \emptyset$\;
	\For{$u\in N_R(v)$}{
		\For{$w \in \LMV{v}$}{
			\If{$uw\in E(\M_v)$}{
				Add directed edge $uw$ to $E(\vec H_v)$\;
			}
			\ElseIf{$uw \in E(G)$}{
				\If{$u\in \RMV{v}$}{
					Add directed edge $wu$ to $E(\vec H_v)$\;
				}
				\Else{
					Add $w$ to $T$\;
					$\out(w) = u$\;
				}
			}
		}
	}	
	\For{$u\in \RMV{v}$}{
		\For{$w\in N_L(u)$}{
			\If{$w\notin \LMV{v} \cup N_L(v) \cup \{v\}$}{
					Add $u$ to $S$\;
					$\out(u) = w$\;
			}
		}
	}
	\tcp{Augment matching if possible}
	\If {there is a directed path $\vec P$ from $S$ to $T$ in~$\vec H_v$}{
		Let~$s \in S$, $t \in T$ be the endpoints of~$\vec P$\;
		$E^+ \defeq E(\vec P) \setminus \M_v \cup \{ s\out(s), t\out(t) \}$\;
		$E^- \defeq E(\vec P) \cap \M_v$\;
		Remove~$E^-$ from~$\M_v$ and add~$E^+$ to~$\M_v$\;
		\Return TRUE\;
	}
	\Else{
		\tcp{Includes the case when one of $S,T$ is empty}
		\Return FALSE\;
	}
	
	\caption{\label{alg:Augmenting}%
		Attempts to increase the local matching $\M_v$ for a given vertex~$v$.}
\end{algorithm2e}

\begin{lemma}\label{lemma:UpdateRT}
	The total running time of \Update is $O(|G|  p^{3})$ plus the total running time of \Augmenting. 
\end{lemma}
\begin{proof}
	We assume that all basic data structures are implemented using hashing, and therefore all operations excluding loops take expected constant time.
	Thus, to bound the total running time we only need to bound the number of times every loop is executed.
		
	Loop~\Mark{1} in \Update is over vertices $N(v)$. Since every vertex of~$G$
	takes the role of~$v$ precisely once, this contributes~$O(\|G\|) = O(p|G|)$ to the total running time.

	Loop~\Mark{2} iterates over the vertices of $N_L(v)$.
	Since $v\in \Candidates$, and Condition~(1) for the data structures holds by Lemma~\ref{lemma:Update}, we know that $\pp_L(v) \leq p$ and in particular 
	(by Lemma~\ref{lemma:small-pp-small-matching}) $|N_L(v)| \leq p$. Consequently the running time of the Loop~\Mark{2} is $O(p^2)$.
	
	Loop~\Mark{3} iterates over $\bigcup_{x \in \RMV{v}} N_L(x)$ and it has an internal loop over the vertices of $N_L(z)$.
	Since $\pp_L(v) \leq p$, by Lemma~\ref{lemma:small-pp-small-matching}, $|\RMV{v}| \leq p$. 
	As every vertex in $\RMV{v}$ is in $R$ and therefore was at some earlier point in $\Candidates$ we have, by the same reasoning as for $v$, that $|N_L(x)|\leq p$, for every $x\in \RMV{v}$, and therefore $|\bigcup_{x \in \RMV{v}} N_L(x)| \leq p^2$.
	We note that prior to the removal of $vz$ from $\M_u$ it held that $v\in \LMV{u}$ and hence $z\in \RMV{u}$. So $z \in R$ which again means that $|N_L(z)|\leq p$ for every $z$ using during the loop.
	Consequently, the total running time of Loop~\Mark{3} is $O(p^3)$.
	
	Finally, Loop~$\Mark{4}$ iterates over all the vertices in $\mathsf{Check}$.
	Since Loop~\Mark{2} and~\Mark{3} add at most one vertex per iteration to~$\mathsf{Check}$, we have that $|\mathsf{Check}| = O(p^3)$ and therefore we have that many iterations of Loop~\Mark{4}.
	
	Finally, since in a full execution of \Main, \Update is called exactly once for every vertex in $G$ the overall running time of \Update is  $O(|G| p^{3})$ plus the total running time of all calls to \Augmenting. 	
\end{proof}

\begin{lemma}\label{lemma:AugmentingCalls}
	In a full execution of \Main, for every $u\in V(G)$, \Update executes \Augmenting with $u$ as an input at most $O(p^2)$ times.
\end{lemma}
\begin{proof}
	Fix an arbitrary vertex $u \in V(G)$ and consider any execution of \Update. Let $\M_u$ and $N_{L}(u)$ be the values $\M_u$ and $N_{L}(u)$ as they are at the beginning. We again annotate data structures with~$\Marki$ to refer to their value just after the execution of Loop~\Mark{i}. 

	Loop~\Mark{4} calls \Augmenting($u$) exactly when
	$|\Three \M_u| + |\Three N_L(u)| = p$ and $u \not\in \Three \Candidates$.
	Note that $\Three \Candidates = \Candidates$ and~$\Three N_L(u) = \One N_L(u)$.
	By Lemma~\ref{lemma:Update}, this can only happen if $|\M_u| + |N_L(u)| = p + 1$ and either (i) Loop~\Mark{1} reduced $N_L(u)$ and Loop~\Mark{2} did not add an edge to~$\M_u$, or (ii) Loop~\Mark{3} removed an edge from~$\M_u$ and did not add another one. We now show that (i) and (ii) can happen at most~$O(p^2)$ times for~$u$ during a full run of \Main.
	
	Consider the first execution of \Update where \Augmenting($u$) was called and let us call this point in time~$\star$. Since $|\starred \M_u| + |\starred N_L(u)| = p$ we know that $|\starred N_L(u)| \leq p$. 
	We further have that~$\Four \starred N_L(u) \subseteq \starred N_L(u)$ so this bound holds for the rest of the run of \Main. Therefore, (i) can only happen $|\starred N_L(u)| \leq p$ many
	times. Now note that every time~(ii) happens, the Loop~\Mark{3} removes an edge from the current matching~$\M_u$ and in particular a vertex from the
	current~$\LMV{u}$. We will now bound the number of vertices that can from $\star$ onwards appear in~$T^2_L(u)$, and therefore in~$\LMV{u}$, which then bound the number of times~(ii) can happen.

	First, let us bound the number of vertices in~$\starred T^2_L(u)$. Since~$\starred M_u$ is a maximal matching between~$\starred N_R(u)$ and~$\starred T^2_L(u)$ it holds that~$\starred T^2_L(u) \subseteq \bigcup_{x \in \starred \RMV{u}} N_L(x)$. 
	Therefore, $|\starred T^2_L(u)| \leq p |\starred \RMV{u}|$ as every vertex in~$\starred R$	has at most~$p$ neighbours in~$\starred L$. As established above,	$|\starred \M_u| \leq p$ therefore $|\starred T^2_L(u)| \leq p^2$.

	Now note that if some vertex~$z$ is added to~$T^2_L(u)$ at some point after $\star$, then it is because some vertex~$v \in N_L(u)$ with~$z \in N_L(v)$ was moved from~$L$ to~$R$ . At this point, $|N_L(v)| \leq p$ so~$v$ adds at most~$p$ vertices to~$T^2_L(u)$. Finally note that this can happen at most~$p$ times since~$|\starred N_L(u)| \leq p$ and this set can only decrease in size. In total, these occurrences after~$\star$ add at most~$p^2$ further vertices to~$T^2_L(u)$. 

	We conclude that in total at most~$O(p^2)$ vertices are added to~$T^2_L(u)$
	after~$\star$ and therefore at most that many vertices can appear in~$\LMV{u}$. Thus~(ii) can only happen~$O(p^2)$ times in total which proves the claim.
\end{proof}

\lemmaRada*
\begin{proof}
	Let~$\M_v$ be the matching before \Augmenting($v$) is executed and $\starred \M_v$ afterwards. As shown in the the proof of Lemma~\ref{lemma:Update}, when \Augmenting($v$) is called, the matching $\M_v$ is maximal with at most $p$ edges.

	Let~$H_v$ be the bipartite graph with sides~$N_R(v)$, $T^2_L(v)$ and all edges between these sets that exist in~$G$. If~$\M_v$ is not maximum, then there exists an augmenting path~$P$ in the~$H_v$ which starts in~$T^2_L(u) \setminus V(\M_v)$, goes through~$H_v[\M_v]$ with alternating matching and non-matching edges, and ends in~$N_R(v) \setminus V(\M_v)$. 

	Let~$\vec H_v$, $S \subseteq N_R(u)$, $T \subseteq T^2_L(v)$, and~$\out$ be as constructed by \Augmenting. Let us first show how $\vec H_v$ relates to~$H_v$:

	\begin{claim}
		The bipartite graph~$\vec H_v[\M_v]$ is an orientation of~$H_v[\M_v]$ where the edges~$\M_v$ go from~$N_R(v)$ to~$T^2_L(v)$ and all other edges go in the other direction. 
	\end{claim}
	\begin{proof}
		The loops in line~4 and line~5 iterate over all pairs~$u \in N_R(v)$, $w \in \LMV{v}$ and add either the arc~$uw$ to~$\vec H_v$ if~$uw \in E(M_v)$ or the arc~$wu$ if~$uw \in E(G) \setminus E(M_v)$. Therefore all edges in~$H_v[\M_v]$ appear in~$\vec H_v$.
	\end{proof}

	\noindent
	We argue that if an augmenting path~$P$ exists in~$H_v$, then there exists a directed
	path~$\vec P$ from~$S$ to~$T$ in $\vec H_v$. To that end, let~$P = asQtb$,
	with~$a \in T^2_L(v)$, $sQt \subseteq V(M_v)$, and~$b \in N_R(v)$. Since~$a \in N_R(s)$ and~$b \in N_L(t)$, we can conclude that~$\out(s)$ and~$\out(t)$ exist and that~$s \in S$ and $t \in T$. 
	Since the edges in~$sQt$ must alternate between matching and non-matching edges, it is easy to verify that~$sQt$ is indeed a directed path in~$\vec H_v$. 
	Accordingly, if \Augmenting returns FALSE the matching $\M_v$ does not have an augmenting path in~$H_v$ and is therefore a maximum matching.

	Let us now argue that for any directed path~$\vec P = s\vec Q t$ from~$S$ to~$T$ in~$\vec H_v$ it holds that the path~$\out(s) s \vec Q t \out(t)$ is an augmenting path for~$\M_v$ in~$H_v$. By construction of~$\vec H_v$, every vertex in~$N_R(v)$ has exactly one outgoing edge, namely the matching edge. Therefore, a directed path from~$S \subseteq N_R(v)$ to~$T \subseteq T^2_L(v)$ must alternative between matching and non-matching edges. Since~$\out(s) s$
	and $t \out(t)$ are edges outside of~$\M_v$, we conclude that $\out(s) s \vec Q t \out(t)$ is indeed an augmenting path for~$\M_v$ in~$H_v$. Lines~$20$ to~$24$ perform the usual augmentation of~$\M_v$ using the edges of this path, therefore we conclude that \Augmenting indeed returns a matching~$\starred M_v$ of size $|\M_v| + 1$. The maximality of~$\starred M_v$ simply follows from the maximality of~$\M_v$.
\end{proof}


\begin{lemma}\label{lemma:AugmentingRT}
	The total runtime of \Augmenting in a full execution of \Main is $O(p^4 |G|)$.
\end{lemma}
\begin{proof}
	We note that since we use hashes for the data structures, all the operations excluding loops take constant time.
	Thus, to bound the total running time we only need to bound the number of times every loop is executed over all calls to \Augmenting.
	
	The first loop in \Augmenting is over the vertices in $N_R(v)$ and it has an internal loop over the vertices of $\LMV{v}$.
	According to \Update, when the algorithm is executed for $v$, $|\LMV{v}| \leq p$.
	By the second condition over the data structures, $N_R(v)\subseteq N(v)$.
	According to Lemma~\ref{lemma:AugmentingCalls}, \Augmenting is executed at most $O(p^2)$ times with $v$ as an input.

	So, over a full execution of \Main, the total running time of the first loop is 
	\[ 
		\sum_{v\in V(G)}O(p^2)|N(v)|p = O(p^3) \cdot 2\|G\| = O(p^4 |G|).
	\]
	The second loop in \Augmenting is over the vertices in $\RMV{v}$ and it has an internal loop over the vertices of $N_L(v)$.
	Recall that \Update calls \Augmenting($v$) when $|N_L(v)| + |\RMV{v}| \leq p$.
	So, by a similar computation to the above, over a full execution of \Main the total running time of the second loop is $O(p^4 |G|)$.
	
	Finally, the conditional statement at line~19 of \Augmenting requires the computation of a directed path. By executing a simple DFS from~$S$ to~$T$
	in~$\vec H_v$ in time~$O( \|\vec H_v\|) = O(p^2)$. The computation of these paths therefore also contributed a total running time of~$O(p^4 |G|)$ and the claim follows.
\end{proof}

\begin{theorem}\label{theorem:main}
	On input graph $G$ and input parameter $p$,
	the total runtime of \Main	is $O(p^4 |G|)$, the total space used \Main is $O(\|G\| + p^2)$,
	 and the algorithm returns FALSE if $\adm_2(G) > p$ and otherwise it returns an ordering $\orderstrict$ such that $\adm_2(G,\orderstrict)\leq p$ of the vertices of $G$.
\end{theorem}
\begin{proof}
	According to Lemma~\ref{lemma:UpdateRT}, the overall running of \Update is $O(p^{3} |G|)$ plus the total running time of \Augmenting.
	According to Lemma~\ref{lemma:AugmentingRT}, the total running time of all the executions of \Augmenting is $O(p^4 |G|)$.
	The operations of \Main do not add more time and we arrive at the claimed running time of $O(p^4 |G|)$.
	
  \noindent
 	To compute the space, recall that the Oracle data structure consists of
 	\begin{itemize}
 		\item $\Candidates$, a subset of $V(G)$, and hence of size at most $|G|$.
 		\item $N_L(v)$ and $N_R(v)$ for every $v\in V(G)$, so overall order of $\sum_{v\in V(G)}|N_L(v)| + |N_R(v)| = \sum_{v\in V(G)}|N(v)| = \|G\|$, and
 		\item A matching $\M_v$, between the vertices of $N_R(v)$ and some other vertices, thus, overall  $\sum_{v\in V(G)} |N_R(v)| \leq \sum_{v\in V(G)}|N(v)| = \|G\|$.
 	\end{itemize} 
 	Beyond this, the routine \Augmenting when executed for a vertex $v$
 	constructs a directed graph between the vertices of $\RMV{v}$ and $\LMV{v}$. Since \Update ensures this only happens when $|\M_v| = O(p)$,
 	we can conclude that \Augmenting uses at most $O(p^2)$ space.
 	Consequently the total space used by \Main is $O(p^2 + \|G\|)$. 
 	
  Finally, to prove the correctness of the algorithm, recall that
  by Lemma~\ref{lemma:r-admissibility} the $2$-admissibility of the input graph
  is decided correctly if every execution of \Update returns a vertex~$v$
  such that~$\pp_L(v) \leq p$ if such a vertex exists and otherwise return FALSE. Lemma~\ref{lemma:Augmenting} proves that \Augmenting works as claimed,
  accordingly Lemma~\ref{lemma:Update} then shows that \Main works as claimed.
\end{proof}

\newcommand{\netw}[1]{\texttt{#1}}

\section{Implementation details and experiments}

We implemented the algorithm in Rust v1.75 and computed the $2$-admissibility of \numnetworks real-world networks. The full results are available in the repository, we present here a summary of our findings.
The dataset contains networks from various domains, like
biology (protein-protein interaction networks like \netw{BioGrid-Homo-Sapiens}, neural networks like \netw{bn-mouse\_retina\_1});
infrastructure networks 
(\netw{autobahn}, \netw{euroroad}, \netw{exnet-water});
social networks (\netw{livemocha}, \netw{soc-gplus}, \netw{twittercrawl});
dependency graphs (\netw{JDK\_dependency}, \netw{linux});
web crawls (\netw{web-EPA}, \netw{web-california});
citation and co-author networks (\netw{cit-HepPh}, \netw{ca-HepPh}, \netw{zewail});
digital communication networks (\netw{email-Enron});
word networks (\netw{foldoc}, \netw{edinburgh\_associative\_thesaurus});
and ownership networks (\netw{bahamas}, \netw{offshore}). 

 All experiments were run with an AMD Ryzen R1600 @ 2.6 GHz and 8GB of RAM.
Since the algorithm is designed to solve the decision variant, we compute the $2$-admissibility using binary search. We measured memory consumption by running separate experiments with a modified memory allocator that records peak allocation. The variation of both running times and peak memory were small, we therefore report results of single runs here and not averages over multiple runs.

\begin{figure}[!tbh]
  \vspace*{-8pt}%
  \hspace*{-10pt}\includegraphics[width=\textwidth]{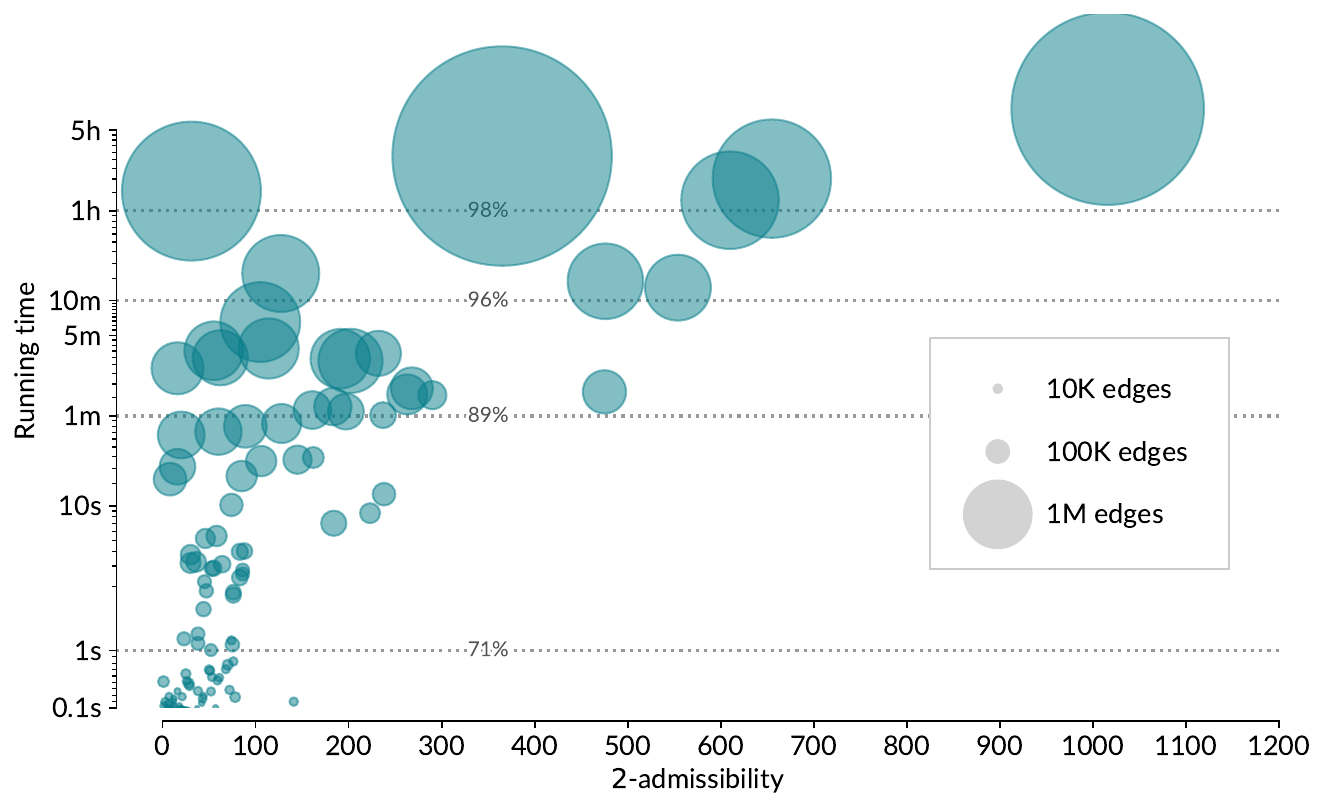}
  \vspace*{-5pt}
  \caption{%
    \label{fig:time}\
    \small
    Running time dependence on $2$-admissibility. Marker sizes represent the number of edges. The dotted lines and percentages indicate on how many networks the algorithm finished below the specified time.
  }
\end{figure}

\begin{figure}[!bth]
  \vspace*{-8pt}%
  \hspace*{-10pt}\includegraphics[width=\textwidth]{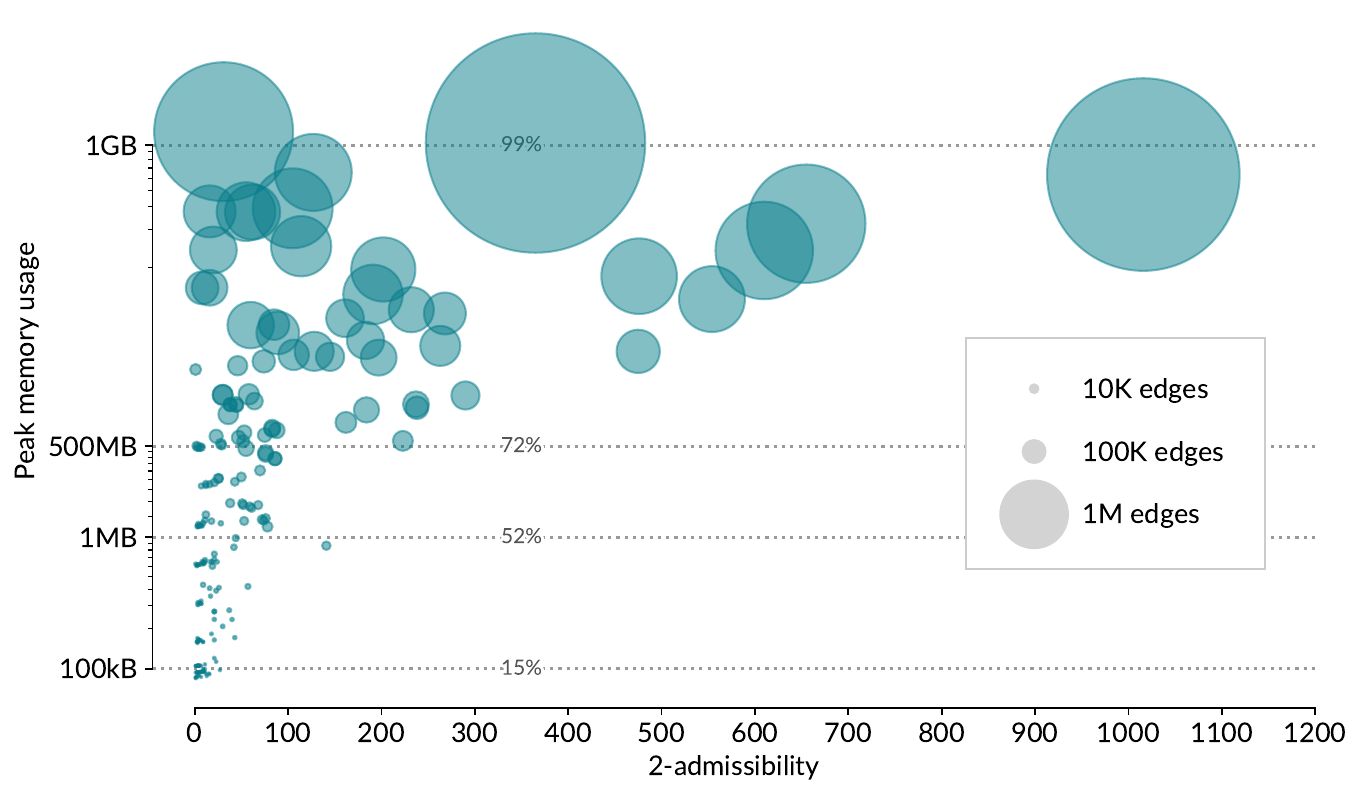}\vspace*{-5pt}
  \caption{%
    \label{fig:memory}\
    \small
    Peak memory consumption dependence on $2$-admissibility. Marker sizes represent the number of edges. The dotted lines and percentages indicate on how many networks the algorithm used less than the indicate amount of memory. On average, a third of the memory consumption is due to storing the network.
  }
\end{figure}

\paragraph*{The algorithm is resource-efficient}

For the running time (Figure~\ref{fig:time}) we find that the algorithm performs well for networks whose $2$-admissibility lies below $\sim$300, in this range the network size is the dominating factor. This indicates that the worst-case of~$O(p^4 |G|)$ is quite pessimistic and therefore that the `optimistic' algorithm design paid off. The computation took less than ten minutes on almost all networks (96\%), including some with millions of edges, and finished in less than a second for most smaller networks (71\%). 

Memory consumption was, as predicted by theory, modest (Figure~\ref{fig:memory}). None of our experiments used more than 1.34 GB of RAM. Note that already simply storing the networks in-memory accounts for a good fraction of the total memory usage: Among the 44~networks with noticable memory footprint (> 10MB), on average 30\% of memory accounts for storing the network itself and 70\% is used by the algorithm's data structures.

\paragraph*{The $2$-admissibility is small in many real-world graphs}

\begin{figure}[!bth]
   \hspace*{-8pt}\includegraphics[width=\columnwidth]{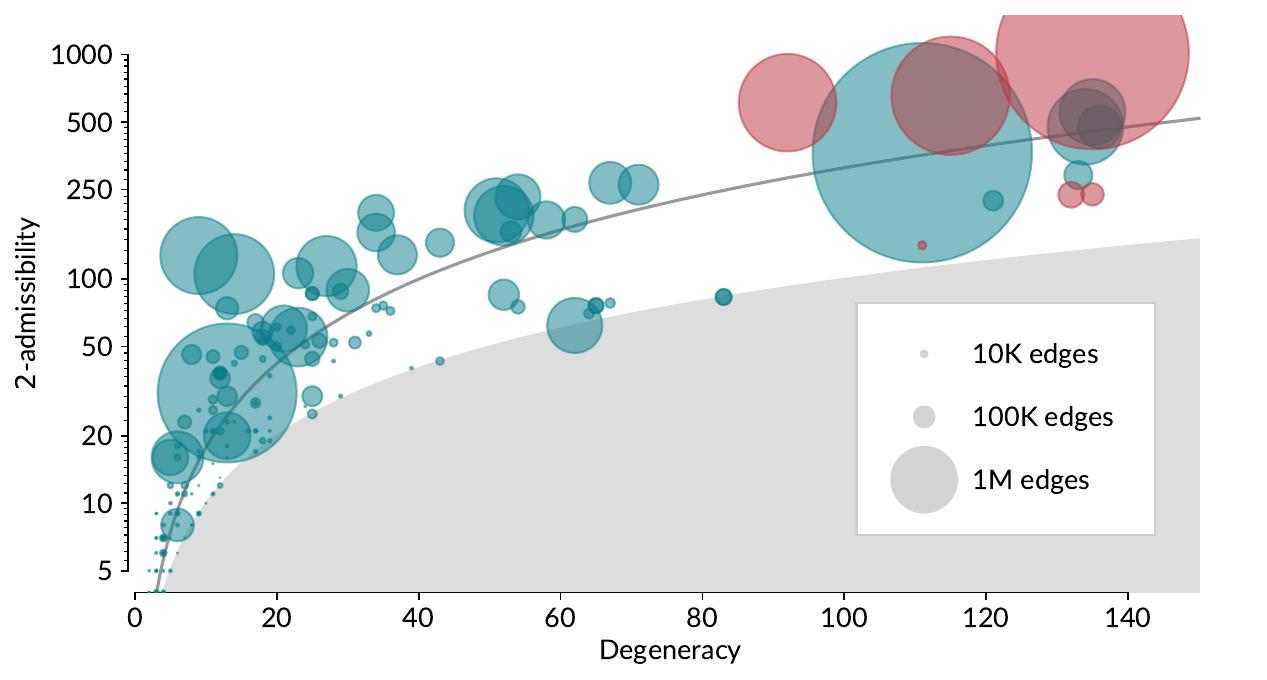}
   \vspace*{-5pt}
  \caption{%
    \label{fig:degen}
    \small
    Comparison of $2$-admissibility and degeneracy. Marker sizes represent the number of edges. The line represents a best-fit polynomial $\adm_2 \approx \deg^{1.25}$, the red markers are outliers whose residual has an
    absolute z-score $\geq 3$.
  }
\end{figure}

We find that 25\% of networks have a $2$-admissibility of less than~$5$ and 50\% have a $2$-admissibility of~$16$ or less (the latter includes networks with over half a million edges). This means that even algorithms that have an exponential dependence on the $2$-admissibility could be useful on these instances.

As networks vary in density, we compare the $2$-admissibility relative to the degeneracy of the networks to obtain a more unified view. Let~$d$ be the degeneracy (the independent variable) and $a_2$ the $2$-admissibility (the dependent variable) of the networks. We fitted the functions $a_2 = \alpha \dot d$, $a_2 = d^\alpha$, and~$a_2 = \alpha^d$ with parameter~$\alpha$ using \texttt{numpy}'s non-linear least squares optimisation, with the following results:

\smallskip
  \hfil{\begin{tabular}{l|lll}
             & Linear  & Polynomial  & Exponential \\ \hline
    Best fit & $3.045 \cdot d$ & $d^{1.2477}$  & $1.0475^d$ \\
    $R^2$    & $0.676$         & $0.693$       & $0.484273$  
  \end{tabular}}\hfil \\
\smallskip

\noindent
Slightly more complicated functions (like $\alpha d^\beta + \gamma$) did not result in a significantly better fit and we chose to keep the modelling 
simple. Hence, we propose that the $2$-admissibility grows roughly as $d^{1.25}$ (\cf Figure~\ref{fig:degen}), as such we expect most networks with low degeneracy to also have low $2$-admissibility

We find six outlier networks with this model by computing the
residual Z-scores. \netw{ca-HepPh}, \netw{mousebrain}, and \netw{twittercrawl}
have a residual Z-score $\leq -3$, therefore their $2$-admissibility is lower than expected; \netw{dogster\_friendships}, \netw{livemocha}, and
\netw{tv\_tropes} have a residual Z-score $\geq 3$, therefore their
$2$-admissibility is higher than expected. One possible reason is that the
first three networks have a relatively high clustering coefficient\footnote{The clustering coefficient of a vertex is the number of triangles it is part of divided by the number of pairs in its neighbourhood; the clustering coefficient of a graph is the average over these values.} (0.61, 0.76, and 0.26, respectively) compared to the latter three (0.17, 0.05, and 0.0, respectively), we leave the investigation of this relationship for future work.

\section{Conclusion}

We demonstrated that careful algorithm design and an `optimistic' approach resulted in a resource-efficient implementation to compute the $2$-admissibility of real-world graphs. Our experiments not only show that the implementation is of practical use, but also that the $2$-admissibility is indeed a measure that is reasonably small on most practical instances. 

In the future, we want to take our learnings from this algorithm and apply it to implement algorithms for $3$- and $4$-admissibility, as these measures might provide further insight into network structure and might still be tractable in practice. 

\bibliography{biblio}

\section{Appendix}

All network data used here is avialable at \url{https://github.com/microgravitas/network-corpus}.

The following pages contains the key sparseness measures as well as the experimental measurements (time/memory) for computing the $2$-admissibility for all \numnetworks networks. The relevant statistics are:
\begin{itemize}
  \item  $\adm_2$: $2$-admissibility,
  \item  $\bar d$: Average degree,
  \item  $\text{deg}$: Degeneracy,
  \item  $\Delta$: Maximum degree, and
  \item  $m$, $n$: Number of edges / number of nodes.
\end{itemize}
The following table summarizes the statistics over all networks:
\begin{center}
\begin{tabular}{@{}lrrrrrr@{}}\small
 & $\adm_2$ & $\bar d$ & $\text{deg}$ & $\Delta$ & $m$ & $n$ \\ \midrule 
mean & 56.87  & 13.02   & 21.20 & 2198.43 & 223391.95  & 52057.52  \\
std  & 119.07 & 26.81   & 30.53 & 9402.44 & 1046710.63 & 248992.77 \\
min  & 1      & 1.48    & 1     & 3.00    & 8          & 9         \\
25\% & 4.5    & 2.74    & 3     & 29.5    & 460        & 186       \\
50\% & 16     & 4.68    & 9     & 102.5   & 3281       & 1367.5    \\
75\% & 54.5   & 13.37   & 24    & 707.5   & 41147.5    & 9049.5    \\
max  & 1016   & 205.26  & 136   & 110602  & 11095298   & 2852295   \\
\end{tabular}
\end{center}

\noindent
The following table summarizes the experimental measurements:
\begin{center}
\begin{tabular}{@{}lrrr@{}}\small
 &  & Peak & Network \\
 &  Time (s) & memory (kB) & memory (kB) \\ \midrule 
mean & 283.73  & 33689.45 & 9724.70 \\
std  & 2130.33 & 140221.88 & 39774.26 \\
min  & 0.02    & 84.46 & 84.46 \\
25\% & 0.03    & 158.15 & 126.48 \\
50\% & 0.08    & 728.77 & 287.30 \\
75\% & 1.90     & 6070.65 & 1843.79 \\
max  & 27739.26 & 1331825.50 & 369765.78 \\
\end{tabular}
\end{center}

\begin{landscape}
\begin{longtable}{@{}>{\small}
    lrrrrrrrrr
    @{}}
& & & & & & &  & Peak & Network \\
Network & $\adm_2$ & $\bar d$ & $\text{deg}$ & $\Delta$ & $m$ & $n$ & Time (s) & memory (kB) & memory (kB) \\ \midrule \endhead
AS-oregon-1 & 28 & 4.19 & 17 & 2389 & 23409 & 11174 & 0.48 & 5343.72 & 1560.89 \\
AS-oregon-2 & 52 & 5.71 & 31 & 2432 & 32730 & 11461 & 1.00 & 5593.12 & 1642.18 \\
BioGrid-Affinity-Capture-Luminescence & 8 & 2.51 & 6 & 376 & 2312 & 1840 & 0.07 & 1252.80 & 445.20 \\
BioGrid-Affinity-Capture-Ms & 183 & 15.90 & 58 & 2217 & 321887 & 40495 & 72.31 & 33281.05 & 9174.07 \\
BioGrid-Affinity-Capture-Rna & 75 & 6.22 & 54 & 3572 & 42815 & 13765 & 1.09 & 6239.51 & 1837.80 \\
BioGrid-Affinity-Capture-Western & 64 & 6.09 & 17 & 535 & 64046 & 21028 & 3.11 & 11368.61 & 3158.82 \\
BioGrid-All & 476 & 34.86 & 134 & 3620 & 1316843 & 75550 & 883.84 & 103501.16 & 28946.62 \\
BioGrid-Arabidopsis-Thaliana-Columbia & 53 & 9.20 & 26 & 1341 & 47916 & 10417 & 2.85 & 6512.89 & 1845.78 \\
BioGrid-Biochemical-Activity & 29 & 4.12 & 11 & 427 & 17746 & 8620 & 0.45 & 5211.81 & 1572.68 \\
BioGrid-Bos-Taurus & 4 & 1.87 & 3 & 27 & 424 & 454 & 0.03 & 308.45 & 170.70 \\
BioGrid-Caenorhabditis-Elegans & 70 & 7.40 & 64 & 522 & 23646 & 6394 & 0.78 & 3335.76 & 978.66 \\
BioGrid-Candida-Albicans-Sc5314 & 9 & 2.87 & 9 & 427 & 1609 & 1121 & 0.06 & 637.82 & 265.64 \\
BioGrid-Canis-Familiaris & 2 & 1.75 & 2 & 90 & 125 & 143 & 0.02 & 105.33 & 105.33 \\
BioGrid-Chemicals & 1 & 1.69 & 1 & 413 & 28093 & 33266 & 0.51 & 19875.67 & 5449.23 \\
BioGrid-Co-Crystal-Structure & 5 & 1.76 & 5 & 92 & 2021 & 2291 & 0.08 & 1246.29 & 432.72 \\
BioGrid-Co-Fractionation & 83 & 10.23 & 83 & 187 & 56354 & 11017 & 2.38 & 6922.47 & 1978.86 \\
BioGrid-Co-Localization & 9 & 2.51 & 6 & 63 & 4452 & 3543 & 0.09 & 1316.80 & 443.30 \\
BioGrid-Co-Purification & 12 & 2.76 & 12 & 1972 & 5970 & 4326 & 0.21 & 2552.52 & 808.07 \\
BioGrid-Cricetulus-Griseus & 1 & 1.65 & 1 & 30 & 57 & 69 & 0.02 & 93.68 & 93.68 \\
BioGrid-Danio-Rerio & 3 & 2.04 & 3 & 61 & 266 & 261 & 0.02 & 157.44 & 127.07 \\
BioGrid-Dictyostelium-Discoideum-Ax4 & 1 & 1.48 & 1 & 4 & 20 & 27 & 0.03 & 85.73 & 85.73 \\
BioGrid-Dosage-Growth-Defect & 9 & 3.03 & 5 & 213 & 2193 & 1447 & 0.04 & 652.36 & 261.90 \\
BioGrid-Dosage-Lethality & 8 & 2.58 & 4 & 392 & 2289 & 1776 & 0.04 & 663.84 & 261.92 \\
BioGrid-Dosage-Rescue & 11 & 3.81 & 7 & 75 & 6444 & 3380 & 0.10 & 1385.21 & 453.51 \\
BioGrid-Drosophila-Melanogaster & 83 & 12.98 & 83 & 303 & 60556 & 9330 & 3.98 & 7080.36 & 2003.70 \\
BioGrid-Emericella-Nidulans-Fgsc-A4 & 2 & 1.94 & 2 & 44 & 62 & 64 & 0.02 & 94.28 & 94.28 \\
BioGrid-Escherichia-Coli-K12-Mg1655 & 10 & 2.97 & 5 & 58 & 1889 & 1273 & 0.07 & 644.16 & 262.50 \\
BioGrid-Escherichia-Coli-K12-W3110 & 290 & 89.40 & 133 & 1187 & 181620 & 4063 & 90.82 & 12557.59 & 3859.89 \\
BioGrid-Far-Western & 3 & 1.82 & 3 & 60 & 1089 & 1199 & 0.03 & 626.99 & 258.45 \\
BioGrid-Fret & 24 & 2.82 & 19 & 51 & 2395 & 1700 & 0.08 & 662.16 & 259.59 \\
BioGrid-Gallus-Gallus & 4 & 2.11 & 4 & 110 & 436 & 413 & 0.02 & 162.87 & 126.99 \\
BioGrid-Glycine-Max & 2 & 1.77 & 2 & 13 & 39 & 44 & 0.03 & 88.31 & 88.31 \\
BioGrid-Hepatitus-C-Virus & 1 & 1.97 & 1 & 133 & 134 & 136 & 0.03 & 105.31 & 105.31 \\
BioGrid-Homo-Sapiens & 263 & 30.69 & 71 & 2882 & 369767 & 24093 & 92.36 & 30177.38 & 9625.87 \\
BioGrid-Human-Herpesvirus-1 & 3 & 2.34 & 3 & 40 & 208 & 178 & 0.03 & 105.47 & 105.47 \\
BioGrid-Human-Herpesvirus-4 & 2 & 2.02 & 2 & 154 & 326 & 323 & 0.03 & 161.15 & 127.41 \\
BioGrid-Human-Herpesvirus-5 & 1 & 1.77 & 1 & 27 & 107 & 121 & 0.02 & 105.18 & 105.18 \\
BioGrid-Human-Herpesvirus-8 & 3 & 1.93 & 3 & 119 & 691 & 716 & 0.04 & 323.39 & 171.72 \\
BioGrid-Human-Immunodeficiency-Virus-1 & 6 & 2.32 & 3 & 324 & 1319 & 1138 & 0.03 & 635.14 & 264.01 \\
BioGrid-Human-Immunodeficiency-Virus-2 & 1 & 1.58 & 1 & 6 & 15 & 19 & 0.05 & 85.88 & 85.88 \\
BioGrid-Human-Papillomavirus-16 & 2 & 2.15 & 2 & 93 & 186 & 173 & 0.04 & 105.64 & 105.64 \\
Cannes2013 & 114 & 3.82 & 27 & 15169 & 835892 & 438089 & 230.83 & 175852.25 & 48981.02 \\
CoW-interstate & 7 & 3.51 & 4 & 25 & 319 & 182 & 0.02 & 104.51 & 104.51 \\
DNC-emails & 28 & 4.70 & 17 & 402 & 4384 & 1866 & 0.06 & 1309.34 & 470.83 \\
EU-email-core & 74 & 32.58 & 34 & 345 & 16064 & 986 & 1.16 & 1376.51 & 455.46 \\
JDK\_dependency & 76 & 16.68 & 65 & 5923 & 53658 & 6434 & 1.86 & 4554.62 & 1408.96 \\
JUNG-javax & 76 & 16.43 & 65 & 5655 & 50290 & 6120 & 1.91 & 4426.94 & 1356.07 \\
NYClimateMarch2014 & 161 & 6.39 & 34 & 14687 & 327080 & 102378 & 67.52 & 49241.35 & 13919.95 \\
NZ\_legal & 68 & 14.70 & 25 & 429 & 15739 & 2141 & 0.71 & 1808.36 & 568.23 \\
Noordin-terror-loc & 4 & 2.99 & 3 & 18 & 190 & 127 & 0.02 & 105.60 & 105.60 \\
Noordin-terror-orgas & 3 & 2.81 & 3 & 21 & 181 & 129 & 0.07 & 105.41 & 105.41 \\
Noordin-terror-relation & 11 & 7.17 & 11 & 28 & 251 & 70 & 0.05 & 94.34 & 94.34 \\
ODLIS & 38 & 11.29 & 12 & 592 & 16377 & 2900 & 0.36 & 1871.51 & 590.77 \\
Opsahl-forum & 42 & 15.65 & 14 & 128 & 7036 & 899 & 0.18 & 857.88 & 348.74 \\
Opsahl-socnet & 61 & 14.57 & 20 & 255 & 13838 & 1899 & 0.57 & 1717.02 & 600.60 \\
StackOverflow-tags & 6 & 4.26 & 6 & 16 & 245 & 115 & 0.07 & 106.36 & 106.36 \\
Y2H\_union & 7 & 2.75 & 4 & 89 & 2705 & 1966 & 0.08 & 1254.80 & 423.44 \\
Yeast & 18 & 6.08 & 6 & 66 & 7182 & 2361 & 0.13 & 1361.24 & 481.72 \\
actor\_movies & 105 & 5.75 & 14 & 646 & 1470404 & 511463 & 390.95 & 343790.00 & 103718.55 \\
advogato & 86 & 15.24 & 25 & 803 & 39285 & 5155 & 2.76 & 4099.85 & 1197.26 \\
airlines & 18 & 11.04 & 13 & 130 & 1297 & 235 & 0.03 & 185.44 & 139.46 \\
american\_revolution & 3 & 2.27 & 3 & 59 & 160 & 141 & 0.03 & 105.50 & 105.50 \\
as-22july06 & 44 & 4.22 & 25 & 2390 & 48436 & 22963 & 1.64 & 10761.58 & 3127.22 \\
as-skitter & 365 & 13.08 & 111 & 35455 & 11095298 & 1696415 & 10768.63 & 1090595.10 & 308145.47 \\
as20000102 & 21 & 3.88 & 12 & 1458 & 12572 & 6474 & 0.28 & 2692.73 & 826.55 \\
autobahn & 3 & 2.56 & 2 & 5 & 478 & 374 & 0.02 & 161.48 & 126.98 \\
bahamas & 8 & 2.24 & 6 & 14902 & 246291 & 219856 & 16.98 & 84164.58 & 22366.35 \\
bergen & 12 & 10.26 & 9 & 32 & 272 & 53 & 0.02 & 91.24 & 91.24 \\
bitcoin-otc-negative & 21 & 4.06 & 16 & 227 & 3259 & 1606 & 0.06 & 695.44 & 270.60 \\
bitcoin-otc-positive & 50 & 6.67 & 20 & 788 & 18591 & 5573 & 0.70 & 2973.39 & 889.74 \\
bn-fly-drosophila\_medulla\_1 & 44 & 10.01 & 18 & 927 & 8911 & 1781 & 0.29 & 1008.70 & 360.61 \\
bn-mouse\_retina\_1 & 223 & 168.79 & 121 & 744 & 90811 & 1076 & 8.61 & 5630.18 & 1957.90 \\
boards\_gender\_1m & 25 & 9.67 & 25 & 88 & 19993 & 4134 & 0.63 & 2885.06 & 952.19 \\
boards\_gender\_2m & 7 & 2.65 & 4 & 45 & 5598 & 4220 & 0.08 & 2534.69 & 809.88 \\
ca-CondMat & 30 & 8.08 & 25 & 279 & 93439 & 23133 & 3.19 & 12682.73 & 3535.77 \\
ca-GrQc & 43 & 5.53 & 43 & 81 & 14484 & 5241 & 0.25 & 2732.29 & 872.88 \\
ca-HepPh & 238 & 19.74 & 135 & 491 & 118489 & 12006 & 12.58 & 10104.79 & 2996.30 \\
capitalist & 21 & 15.41 & 19 & 91 & 1071 & 139 & 0.04 & 120.71 & 114.86 \\
celegans & 21 & 14.46 & 10 & 134 & 2148 & 297 & 0.05 & 239.56 & 146.60 \\
chess & 88 & 15.31 & 29 & 181 & 55899 & 7301 & 4.04 & 6781.29 & 2210.87 \\
chicago & 1 & 1.77 & 1 & 12 & 1298 & 1467 & 0.03 & 639.14 & 260.24 \\
cit-HepPh & 89 & 24.37 & 30 & 846 & 420877 & 34546 & 48.77 & 37981.65 & 10494.86 \\
cit-HepTh & 128 & 25.37 & 37 & 2468 & 352285 & 27769 & 51.60 & 27419.54 & 7700.24 \\
codeminer & 5 & 2.80 & 4 & 55 & 1015 & 724 & 0.06 & 324.96 & 171.79 \\
columbia-mobility & 9 & 9.61 & 9 & 228 & 4147 & 863 & 0.05 & 441.44 & 211.37 \\
columbia-social & 19 & 17.90 & 18 & 545 & 7724 & 863 & 0.09 & 616.05 & 267.89 \\
cora\_citation & 30 & 7.70 & 13 & 377 & 89157 & 23166 & 3.76 & 12701.12 & 3480.62 \\
countries & 16 & 2.11 & 6 & 110602 & 624402 & 592414 & 155.42 & 325340.25 & 92037.16 \\
cpan-authors & 17 & 5.03 & 9 & 327 & 2112 & 839 & 0.05 & 361.74 & 184.45 \\
deezer & 60 & 18.26 & 21 & 420 & 498202 & 54573 & 43.88 & 43563.94 & 12043.83 \\
digg & 46 & 5.68 & 8 & 285 & 86312 & 30398 & 5.19 & 21257.37 & 6488.16 \\
diseasome & 11 & 3.86 & 11 & 84 & 2738 & 1419 & 0.16 & 662.72 & 267.06 \\
dogster\_friendships & 1016 & 40.05 & 135 & 46505 & 8546581 & 426820 & 27739.26 & 624314.94 & 194447.22 \\
dolphins & 6 & 5.13 & 4 & 12 & 159 & 62 & 0.03 & 94.85 & 94.85 \\
dutch-textiles & 5 & 3.75 & 5 & 31 & 90 & 48 & 0.06 & 88.75 & 88.75 \\
ecoli-transcript & 5 & 2.73 & 3 & 74 & 578 & 423 & 0.04 & 165.20 & 127.63 \\
edinburgh\_associative\_thesaurus & 197 & 25.69 & 34 & 1062 & 297094 & 23132 & 65.53 & 24529.89 & 6780.44 \\
email-Enron & 145 & 10.02 & 43 & 1383 & 183831 & 36692 & 25.06 & 24888.59 & 7480.81 \\
escorts & 45 & 4.67 & 11 & 305 & 39044 & 16730 & 2.19 & 10530.04 & 3208.14 \\
euroroad & 3 & 2.41 & 2 & 10 & 1417 & 1174 & 0.07 & 625.63 & 258.82 \\
eva-corporate & 4 & 1.85 & 3 & 552 & 6711 & 7253 & 0.11 & 4950.11 & 1432.23 \\
exnet-water & 3 & 2.55 & 2 & 10 & 2416 & 1893 & 0.05 & 1231.98 & 433.43 \\
facebook-links & 191 & 25.64 & 52 & 1098 & 817090 & 63731 & 188.07 & 74939.01 & 20580.54 \\
foldoc & 36 & 13.70 & 12 & 728 & 91471 & 13356 & 3.25 & 9015.71 & 2506.90 \\
foodweb-caribbean & 23 & 13.47 & 13 & 196 & 3313 & 492 & 0.08 & 397.60 & 201.50 \\
foodweb-otago & 23 & 11.80 & 14 & 45 & 832 & 141 & 0.03 & 113.57 & 112.90 \\
football & 11 & 10.66 & 8 & 12 & 613 & 115 & 0.05 & 108.00 & 108.00 \\
google+ & 38 & 3.32 & 12 & 2761 & 39194 & 23628 & 1.11 & 10679.88 & 3099.64 \\
gowalla & 202 & 9.67 & 51 & 14730 & 950327 & 196591 & 180.06 & 116517.97 & 31640.44 \\
haggle & 40 & 15.50 & 39 & 101 & 2124 & 274 & 0.04 & 239.14 & 148.30 \\
hex & 4 & 5.62 & 3 & 6 & 930 & 331 & 0.03 & 167.64 & 130.15 \\
hypertext\_2009 & 43 & 38.87 & 28 & 98 & 2196 & 113 & 0.06 & 173.63 & 132.26 \\
ia-email-univ & 21 & 9.62 & 11 & 71 & 5451 & 1133 & 0.09 & 762.10 & 293.13 \\
ia-infect-dublin & 21 & 13.49 & 17 & 50 & 2765 & 410 & 0.08 & 276.06 & 157.06 \\
ia-reality & 12 & 2.26 & 5 & 261 & 7680 & 6809 & 0.25 & 2635.98 & 803.31 \\
infectious & 21 & 13.49 & 17 & 50 & 2765 & 410 & 0.06 & 274.55 & 157.04 \\
ingredients & 475 & 197.46 & 136 & 3426 & 431654 & 4372 & 97.07 & 27386.90 & 10230.61 \\
iscas89-s1196 & 4 & 2.85 & 2 & 16 & 537 & 377 & 0.03 & 162.96 & 127.80 \\
iscas89-s1238 & 5 & 3.00 & 2 & 18 & 625 & 416 & 0.02 & 167.24 & 127.98 \\
iscas89-s13207 & 6 & 2.73 & 4 & 37 & 3406 & 2492 & 0.05 & 1273.01 & 440.71 \\
iscas89-s1423 & 3 & 2.62 & 2 & 17 & 554 & 423 & 0.03 & 163.58 & 127.24 \\
iscas89-s1488 & 7 & 3.37 & 3 & 53 & 779 & 463 & 0.03 & 315.73 & 174.40 \\
iscas89-s1494 & 7 & 3.37 & 3 & 56 & 796 & 473 & 0.03 & 316.10 & 174.40 \\
iscas89-s15850 & 4 & 2.47 & 4 & 25 & 4004 & 3247 & 0.05 & 1296.03 & 418.71 \\
iscas89-s27 & 1 & 1.78 & 1 & 3 & 8 & 9 & 0.02 & 84.46 & 84.46 \\
iscas89-s298 & 3 & 2.85 & 2 & 11 & 131 & 92 & 0.02 & 94.33 & 94.33 \\
iscas89-s344 & 3 & 2.44 & 2 & 9 & 122 & 100 & 0.02 & 93.55 & 93.55 \\
iscas89-s349 & 3 & 2.49 & 2 & 9 & 127 & 102 & 0.02 & 94.15 & 94.15 \\
iscas89-s35932 & 2 & 2.55 & 2 & 1440 & 15961 & 12515 & 0.13 & 5156.41 & 1496.98 \\
iscas89-s382 & 4 & 2.90 & 2 & 18 & 168 & 116 & 0.02 & 105.42 & 105.42 \\
iscas89-s38417 & 6 & 2.24 & 4 & 39 & 10635 & 9500 & 0.17 & 5017.87 & 1488.70 \\
iscas89-s38584 & 7 & 2.74 & 4 & 54 & 12573 & 9193 & 0.28 & 5036.81 & 1503.33 \\
iscas89-s386 & 4 & 3.51 & 3 & 23 & 200 & 114 & 0.03 & 106.01 & 106.01 \\
iscas89-s400 & 4 & 3.01 & 2 & 19 & 182 & 121 & 0.03 & 105.44 & 105.44 \\
iscas89-s444 & 4 & 3.07 & 2 & 19 & 206 & 134 & 0.04 & 105.50 & 105.50 \\
iscas89-s510 & 4 & 2.92 & 2 & 12 & 251 & 172 & 0.03 & 105.32 & 105.32 \\
iscas89-s526 & 4 & 3.38 & 3 & 12 & 270 & 160 & 0.02 & 105.60 & 105.60 \\
iscas89-s526n & 4 & 3.37 & 3 & 12 & 268 & 159 & 0.03 & 105.61 & 105.61 \\
iscas89-s5378 & 5 & 2.32 & 3 & 10 & 1639 & 1411 & 0.06 & 633.07 & 257.70 \\
iscas89-s641 & 4 & 2.88 & 3 & 12 & 144 & 100 & 0.04 & 93.57 & 93.57 \\
iscas89-s713 & 4 & 2.63 & 3 & 12 & 180 & 137 & 0.03 & 105.07 & 105.07 \\
iscas89-s820 & 9 & 4.02 & 3 & 48 & 480 & 239 & 0.04 & 160.27 & 129.70 \\
iscas89-s832 & 9 & 4.07 & 3 & 49 & 498 & 245 & 0.04 & 161.36 & 129.52 \\
iscas89-s9234 & 4 & 2.39 & 4 & 18 & 2370 & 1985 & 0.06 & 1243.79 & 419.55 \\
iscas89-s953 & 3 & 2.73 & 2 & 12 & 454 & 332 & 0.03 & 160.63 & 127.39 \\
jazz & 30 & 27.70 & 29 & 100 & 2742 & 198 & 0.05 & 211.72 & 141.45 \\
karate & 4 & 4.59 & 4 & 17 & 78 & 34 & 0.03 & 88.95 & 88.95 \\
lederberg & 47 & 9.98 & 15 & 1103 & 41532 & 8324 & 1.93 & 5968.18 & 1852.88 \\
lesmiserables & 9 & 6.60 & 9 & 36 & 254 & 77 & 0.04 & 94.71 & 94.71 \\
link-pedigree & 2 & 2.51 & 2 & 14 & 1125 & 898 & 0.03 & 615.70 & 258.84 \\
linux & 106 & 13.83 & 23 & 9338 & 213217 & 30834 & 24.34 & 25761.40 & 8177.22 \\
livemocha & 610 & 42.13 & 92 & 2980 & 2193083 & 104103 & 4456.76 & 163105.27 & 49622.93 \\
loc-brightkite\_edges & 85 & 7.35 & 52 & 1134 & 214078 & 58228 & 18.08 & 43914.91 & 13050.88 \\
location & 16 & 2.61 & 5 & 12189 & 293697 & 225486 & 21.74 & 84379.95 & 23410.76 \\
mag\_geology\_coauthor & 31 & 3.12 & 13 & 1153 & 4448428 & 2852295 & 5345.90 & 1331825.50 & 369765.78 \\
marvel & 58 & 9.95 & 18 & 1625 & 96662 & 19428 & 5.46 & 12827.92 & 3554.43 \\
mg\_casino & 9 & 5.98 & 9 & 94 & 326 & 109 & 0.03 & 96.15 & 96.15 \\
mg\_forrestgump & 8 & 5.77 & 8 & 89 & 271 & 94 & 0.03 & 95.35 & 95.35 \\
mg\_godfatherII & 8 & 5.62 & 8 & 34 & 219 & 78 & 0.02 & 94.38 & 94.38 \\
mg\_watchmen & 7 & 5.29 & 7 & 33 & 201 & 76 & 0.03 & 94.48 & 94.48 \\
minnesota & 3 & 2.50 & 2 & 5 & 3303 & 2642 & 0.04 & 1257.38 & 431.87 \\
moreno\_health & 12 & 8.24 & 7 & 27 & 10455 & 2539 & 0.12 & 1525.81 & 484.67 \\
mousebrain & 141 & 151.07 & 111 & 205 & 16089 & 213 & 0.20 & 882.01 & 443.32 \\
movielens\_1m & 554 & 205.26 & 135 & 3428 & 1000209 & 9746 & 777.59 & 69027.36 & 22602.53 \\
movies & 5 & 3.80 & 3 & 19 & 192 & 101 & 0.03 & 94.19 & 94.19 \\
muenchen-bahn & 3 & 2.59 & 2 & 13 & 578 & 447 & 0.03 & 172.26 & 126.32 \\
munin & 3 & 2.11 & 3 & 66 & 1397 & 1324 & 0.04 & 631.20 & 258.60 \\
netscience & 19 & 3.75 & 19 & 34 & 2742 & 1461 & 0.11 & 659.91 & 263.98 \\
offshore & 20 & 3.63 & 13 & 37336 & 505965 & 278877 & 41.04 & 164611.63 & 45987.69 \\
openflights & 52 & 10.67 & 28 & 242 & 15677 & 2939 & 0.36 & 1808.73 & 588.65 \\
p2p-Gnutella04 & 23 & 7.35 & 7 & 103 & 39994 & 10876 & 1.18 & 6104.80 & 1691.07 \\
panama & 62 & 2.52 & 62 & 7015 & 702437 & 556686 & 191.64 & 321508.06 & 90229.39 \\
paradise & 55 & 2.93 & 23 & 35359 & 794545 & 542102 & 220.41 & 324180.06 & 92879.70 \\
photoviz\_dynamic & 7 & 3.24 & 4 & 29 & 610 & 376 & 0.02 & 165.08 & 127.98 \\
pigs & 3 & 2.41 & 2 & 39 & 592 & 492 & 0.05 & 310.38 & 171.13 \\
polblogs & 72 & 27.31 & 36 & 351 & 16715 & 1224 & 0.38 & 1401.45 & 470.51 \\
polbooks & 9 & 8.40 & 6 & 25 & 441 & 105 & 0.05 & 97.88 & 97.88 \\
pollination-carlinville & 53 & 20.34 & 18 & 157 & 15255 & 1500 & 0.58 & 1365.70 & 453.09 \\
pollination-daphni & 26 & 7.36 & 9 & 124 & 2933 & 797 & 0.08 & 419.46 & 198.24 \\
pollination-tenerife & 6 & 3.79 & 4 & 17 & 129 & 68 & 0.04 & 94.84 & 94.84 \\
pollination-uk & 76 & 33.97 & 35 & 256 & 16712 & 984 & 0.83 & 1431.23 & 460.78 \\
ratbrain & 78 & 91.57 & 67 & 497 & 23030 & 503 & 0.27 & 1234.35 & 550.80 \\
reactome & 184 & 46.64 & 62 & 855 & 147547 & 6327 & 7.05 & 9729.96 & 3316.78 \\
residence\_hall & 21 & 16.95 & 11 & 56 & 1839 & 217 & 0.04 & 166.53 & 127.13 \\
rhesusbrain & 37 & 25.24 & 19 & 111 & 3054 & 242 & 0.10 & 281.81 & 160.76 \\
roget-thesaurus & 11 & 7.22 & 6 & 28 & 3648 & 1010 & 0.04 & 684.80 & 281.48 \\
seventh-graders & 16 & 17.24 & 13 & 28 & 250 & 29 & 0.03 & 90.82 & 90.82 \\
slashdot\_threads & 74 & 4.60 & 13 & 2915 & 117378 & 51083 & 10.16 & 22970.25 & 6330.25 \\
soc-Epinions1 & 268 & 10.69 & 67 & 3044 & 405740 & 75879 & 104.16 & 53521.95 & 16119.82 \\
soc-Slashdot0811 & 232 & 12.13 & 54 & 2539 & 469180 & 77360 & 209.41 & 57199.10 & 16037.10 \\
soc-advogato & 86 & 15.26 & 25 & 807 & 39432 & 5167 & 2.56 & 4113.99 & 1200.26 \\
soc-gplus & 38 & 3.32 & 12 & 2761 & 39194 & 23628 & 1.26 & 10679.76 & 3099.64 \\
soc-hamsterster & 51 & 13.71 & 24 & 273 & 16630 & 2426 & 0.68 & 1860.52 & 609.35 \\
soc-wiki-Vote & 16 & 6.56 & 9 & 102 & 2914 & 889 & 0.07 & 415.82 & 194.76 \\
sp\_data\_school\_day\_2 & 57 & 46.55 & 33 & 88 & 5539 & 238 & 0.11 & 429.13 & 200.84 \\
teams & 127 & 2.92 & 9 & 2671 & 1366466 & 935591 & 1031.76 & 648392.60 & 180105.19 \\
train\_bombing & 10 & 7.59 & 10 & 29 & 243 & 64 & 0.02 & 96.34 & 96.34 \\
tv\_tropes & 655 & 42.50 & 115 & 12400 & 3232134 & 152093 & 6861.11 & 261756.61 & 80649.22 \\
twittercrawl & 237 & 84.70 & 132 & 1084 & 154824 & 3656 & 60.97 & 10744.15 & 3414.48 \\
ukroad & 3 & 2.53 & 3 & 5 & 15641 & 12378 & 0.20 & 5052.68 & 1392.84 \\
unicode\_languages & 7 & 2.89 & 4 & 141 & 1255 & 868 & 0.03 & 332.53 & 173.08 \\
wafa-ceos & 7 & 7.15 & 5 & 22 & 93 & 26 & 0.04 & 86.51 & 86.51 \\
wafa-eies & 27 & 28.98 & 24 & 44 & 652 & 45 & 0.03 & 97.84 & 97.84 \\
wafa-hightech & 13 & 15.14 & 12 & 20 & 159 & 21 & 0.02 & 87.76 & 87.76 \\
wafa-padgett & 3 & 3.60 & 3 & 8 & 27 & 15 & 0.02 & 85.95 & 85.95 \\
web-EPA & 16 & 4.17 & 6 & 175 & 8909 & 4271 & 0.36 & 2611.49 & 802.68 \\
web-california & 26 & 5.17 & 11 & 199 & 15969 & 6175 & 0.52 & 2909.39 & 859.52 \\
web-google & 17 & 4.27 & 17 & 59 & 2773 & 1299 & 0.05 & 665.56 & 264.85 \\
wiki-vote & 162 & 28.32 & 53 & 1065 & 100762 & 7115 & 26.17 & 7802.97 & 2226.58 \\
wikipedia-norm & 59 & 16.34 & 22 & 455 & 15372 & 1881 & 0.53 & 1764.65 & 594.99 \\
win95pts & 3 & 2.26 & 2 & 9 & 112 & 99 & 0.02 & 94.17 & 94.17 \\
windsurfers & 15 & 15.63 & 11 & 31 & 336 & 43 & 0.03 & 92.03 & 92.03 \\
word\_adjacencies & 11 & 7.59 & 6 & 49 & 425 & 112 & 0.04 & 99.66 & 99.66 \\
zewail & 55 & 16.29 & 18 & 331 & 54182 & 6651 & 2.87 & 4931.46 & 1441.94 \\ \midrule
\end{longtable}

\end{landscape}

\end{document}